\documentclass[sn-mathphys,Numbered]{sn-jnl}% Math and Physical Sciences Reference Style
\usepackage{graphicx}
\usepackage{multirow}
\usepackage{amsmath,amssymb,amsfonts,natbib}
\usepackage{longtable}
\usepackage{amsthm}%
\usepackage{mathrsfs}%
\usepackage[title]{appendix}%
\usepackage{xcolor}%
\usepackage{textcomp}%
\usepackage{manyfoot}%
\usepackage{booktabs}%
\usepackage{algorithm}%
\usepackage{algorithmicx}%
\usepackage{algpseudocode}%
\usepackage{listings}%
\usepackage{slashbox}

\setcitestyle{authoryear,open={(},close={)}}

\makeatletter
\def\thm@space@setup{%
	\thm@preskip=7pt
	\thm@postskip=\thm@preskip % or whatever, if you don't want them to be equal
}
\makeatother

\theoremstyle{plain}%
\newtheorem{theorem}{Theorem}% 
\newtheorem*{example}{Example}%

\newtheorem{prob}{Problem}
\newtheorem{prop}{Proposition}
\newtheorem{corol}{Corollary}
\newtheorem{lemma}{Lemma}

\numberwithin{equation}{section}

\newcommand{\Real}{\mathbb{R}}

\begin{document}

\title[Optimizing antimicrobial treatment schedules]{Optimizing antimicrobial treatment schedules: some fundamental analytical results}


\author{\fnm{Guy} \sur{Katriel}}\email{katriel@braude.ac.il}

%\author[2,3]{\fnm{Second} %\sur{Author}}\email{iiauthor@gmail.com}


\affil{\orgdiv{Department of Applied Mathematics}, \orgname{Braude College of Engineering}, \orgaddress{\city{Karmiel}, \country{Israel}}}


\abstract{This work studies fundamental questions 
	regarding the optimal design 
	of antimicrobial treatment protocols,
	using pharmacodynamic and pharmacokinetic mathematical models.
	We consider the problem of designing
	an antimicrobial treatment schedule
	to achieve eradication of a microbial 
	infection, while minimizing the area under the time-concentration curve (AUC), which is equivalent to minimizing the cumulative dosage.
	We first solve this problem under the assumption that an arbitrary antimicrobial concentration profile may be chosen, and prove that the {\it{ideal}} concentration profile
	consists of a constant concentration over a finite time duration, where explicit expressions for the optimal 
	concentration and the time duration are given in terms of the pharmacodynamic parameters. 
	Since antimicrobial concentration profiles are induced by a dosing schedule and the antimicrobial pharmacokinetics, the `ideal'
	concentration profile is 
	not strictly feasible. We therefore also 
	investigate the possibility of achieving 
	outcomes which are close to those 
	provided by the `ideal' concentration profile,
	using a bolus+continuous dosing schedule, which consists of 
	a loading dose followed by  infusion of the antimicrobial at a constant rate.
	We explicitly find the optimal bolus+continuous
	dosing schedule, and show that, for realistic parameter ranges, this schedule
	achieves results which are nearly as 
	efficient as those attained by the `ideal'
	concentration profile. The  optimality results obtained here provide a baseline and reference point for comparison and evaluation of antimicrobial treatment plans.
}




\maketitle

\section{Introduction}

Antimicrobial agents have made an immense contribution to human welfare, and 
their effective and efficient use is an issue of crucial importance (\cite{owens,rotschafer}),  particularly in view of the global antimicrobial resistance crisis, which is driven, in part, by 
mis-use or over-use (\cite{murray,ventola}).
Mathematical modelling plays an important role in exploring the dynamics of microbial 
growth, antimicrobial pharmacokinetics (the absorption, distribution and elimination of the drug in the body) and
pharmacodynamics (the drug's effect on the microbial population) (\cite{nielsen,vinks}).
Coupled with
experimental laboratory work and clinical studies, mathematical modelling aids in the design and evaluation of treatment protocols 
and guidelines (\cite{bulitta,rao,rayner}).

A traditional and widely-employed 
approach to the quantitative design of antimicrobial 
treatment regimens employs several 
PK/PD indices which quantify exposure 
over a time period, and uses experimental studies to determine the index 
which is maximally correlated to measures of 
efficacy for a particular antimicrobial, with respect to a specific microbial species (\cite{onufrak,owens,vinks}). A different
methodology, known as `mechanism-based' or `semi-mechanistic' modelling (\cite{bouvier,czock,mi,mueller,nielsen,rao}) relies on modelling the full time-course of treatment using dynamic
models which describe the time dependence of both the microbial population and the antimicrobial agent's concentration, most often using differential equations. Such models include both pharmacokinetic 
parameters, related to drug distribution and elimination, and pharmacodynamic parameters, 
related to antimicrobial effect on the microbial population, and these parameters are estimated by fitting models to experimental data (\cite{bhagunde,czock,kesisoglou,mouton,nielsen,regoes,wen}). Once such a model is calibrated and validated,
it can serve as an {\it{in-silico}} experimental system, allowing to test the outcomes of a variety of treatment schedules. 

Alongside the development of relevant mathematical models, a variety analytical tools have been developed in order to gain understanding of the dynamic behavior of PK/PD models, and of how 
the parameters involved affect the outcomes of treatment (\cite{krzyzanski,macheras,mudassar,nguyen,peletier,rescigno,wu}).

The availability of mechanism-based  models raises the prospect of systematic 
determination of {\it{optimal}} treatment plans,
using mathematical and computational approaches, and indeed several researchers have undertaken such investigations. The dynamical models used, as well as the class of candidate treatment schedules considered and the quantities targeted for optimization, vary among different works. Computationally intensive methods are used for the purpose of finding the optimal schedules, including optimal control methods (\cite{ali,khan,pena,zilonova}), genetic algorithms (\cite{cicchese,colin,goranova,hoyle,paterson}) and machine learning (\cite{smith}). 
While such computational work is very valuable
and has the advantage of enabling the study 
of relatively elaborate models, it is also important to approach antimicrobial treatment 
optimization from an analytical point of view, with the aim of obtaining general insights and mathematical results.  

By employing simple, but widely used, mathematical models, and by formulating  natural optimization problems, we can mathematically prove several general
results characterizing optimal treatment plans. 
An analytic approach provides generic results which are valid for {\it{all}} parameter values of a model, rather than for specific sets of parameters as in numerical studies, and 
enables to obtain useful explicit formulas for determining 
the quantitites characterizing the optimal treatment regimens.
The results yield fundamental 
understanding of the problem of optimal treatment
with an antimicrobial agent. To the extent that the (standard) mathematical models used here capture the dynamics of microbial growth and the effect of antimicrobials, 
the results offer practical guidelines for the 
design of antimicrobial treatment schedules, as 
will be discussed.

We now provide an overview of the contributions presented 
in this work, referring to the corresponding 
sections for details.

We formulate and address some 
key issues regarding optimal 
antimicrobial treatment. Stated simply (and to be formulated more precisely below), our question is: 

\begin{itemize}
	\item[] Given that we aim to use an antimicrobial agent to eradicate a microbial infection, what is the treatment plan that will do so using a minimal cumulative dosage of antimicrobial? 
\end{itemize}

The essential tradeoff underlying this optimization 
problem is that, while a low concentration of 
antimicrobial will be insufficient to suppress 
microbial growth, a very high concentration will 
be wasteful due to the saturation of the 
antimicrobial effect at high concentrations, as well as increase the risk of toxicity. To quantitatively illuminate this tradeoff, we use a standard pharmacodynamic model describing the growth of a microbial population and a killing  rate of microbes depending on the antimicrobial concentration, see section \ref{themodel}. In the context of such a model, the microbial population size cannot reach $0$, so 
that `eradication' is defined in terms of reducing the microbial population size by a given factor, which will depend on the initial microbial  population size.

Our work consists of two parts, in which we address the above question on two levels. In our first investigation (sections \ref{optall}-\ref{hilf}) we allow an arbitrary time-dependence of the antimicrobial concentration at the infection site, and seek to find, among those concentration profiles which lead to eradication of the microbial population, the one for which the area under the time-concentration curve ($AUC$) 
is minimal. The $AUC$ is a standard measure for the overall exposure (\cite{nielsen}), and indeed it is proportional to the cumulative antimicrobial dosage (see equation \eqref{aucd}).
In this formulation of the problem, we are focusing on the pharmacodynamics, ignoring the fact that not every concentration profile is {\it{pharamacokinetically feasible}}, in the sense that it can be induced by an appropriate dosing schedule - these pharmacokinetic aspects are addressed in the second part of the paper. In this general context, we obtain several results:  
\begin{itemize}
	\item[(i)] The optimal antimicrobial concentration profile
	consists of a {\it{constant}} concentration $c_{opt}$ applied for a finite time duration $T_{opt}$.
	\item[(ii)] We find an algebraic equation which allows us to determine the values $c_{opt}$ and $T_{opt}$. In the case that the pharmacodynamics is described by a Hill function (the most commonly employed pharmacodynamic model) we solve this equation
	to obtain explicit expressions for $c_{opt}$ and $T_{opt}$ in terms of the pharmacodynamic parameters (see section \ref{hilf}).
	\item[(iii)] The optimal antimicrobial concentration $c_{opt}$ is independent of the initial size of the microbial population, which only affects the duration $T_{opt}$ of the optimal treatment.
\end{itemize}
The above results establish a {\it{baseline}} in the sense that they provide a lower bound for the $AUC$
needed to achieve eradication. However, 
this analysis focuses only on the pharamacodynamics,
that is the drug effect, ignoring the limitations 
on the concentration profile induced by pharmacokinetics - the dynamics of drug absorption and
elimination.
The  `ideal' concentration curve which achieves the lower bound, consisting of a constant concentration value over a
finite time-interval, is not strictly achievable 
in practice, due to the simple fact that drug 
concentration cannot drop to $0$ in a single instant, but rather decays in a gradual way.
Therefore, in section \ref{pharmacokinetic}, we address the question of achieving efficient treatment using 
drug concentration profiles which are {\it{`pharmacokinetically feasible'}}. 
We would like to achieve results which are close to
the `ideal' baseline determined in the first part 
of this work, but which can be realistically attained by a dosage plan, preferably one that is simple to implement.
In this work we restrict ourselves to a simple one-compartment pharmacokinetic model - leaving 
consideration of more complex pharmacokinetics to 
future work. The only pharmacokinetic parameter
is thus the rate of drug decay.
In this context, we examine simple dosage 
plans of the {\it{bolus+continuous}} type (\cite{derendorf}), in which a single (bolus) dose of the antimicrobial is given at the initiation of treatment, in order to instantaneously raise the drug concentration 
to a level $\bar{c}$, and constant-rate infusion is provided thereafter, for a time duration $T_{bc}$, in order to maintain the same concentration. The 
initial dose and the constant rate of infusion are 
determined by the desired concentration and the
pharmacokinetic parameter (rate of drug decay).
This choice of dosing schedule mimicks 
the `ideal' concentration profile in that the concentration is constant for a finite duration, but with exponential decay thereafter.
Optimizing over all 
such dosing schedules (that is over all choices of $\bar{c}$ and $T_{bc}$) which achieve eradication, with the aim of minimizing the $AUC$ (which is equivalent to minimizing cumulative dosage), we find the following:
\begin{itemize}
	\item[(i)] The optimal concentration $\bar{c}$ is (somewhat surprisingly) identical to the value $c_{opt}$ obtained 
	for the `ideal' concentration profile in the first part of our work. In particular, it does not depend on the pharamacokinetics, that is on the rate of decay of the antimicrobial.
	
	\item[(ii)] The optimal time duration $T_{bc,opt}$ over which 
	the constant-rate infusion of drug should be performed is 
	given by an explicit formula, and depends both on the pharmacodynamics and on the rate of decay of the antimicrobial. This duration is always shorter than the duration $T_{opt}$ of the `ideal' concentration profile. 
	
	\item[(iii)] If the antimicrobial decay rate is sufficiently small, then $T_{bc}=0$, that is the optimal bolus+continuous schedule consists only of a bolus dose, and if the 
	antimicrobial decay rate is large, then $T_{bc}$ is close to $T_{opt}$, and the 
	$AUC$ corresponding to the optimal bolus+continuous treatment is close to (though somewhat higher than) the $AUC$ of the `ideal' concentration profile.
\end{itemize}

Numerical results given in section \ref{bchill}, computed for the case of a Hill-type pharmacodynamic function, with realistic ranges of values of the pharmacokinetic and pharmacodynamic parameters, show that, in most cases, the optimal
bolus+continuous dosing schedule achieves results which are 
nearly as efficient as those attained using the `ideal' concentration profile, in that the $AUC$ valued attained is not significantly higher. We thus conclude that a bolus+continuous dosing schedule, suitably designed,
provides a nearly-optimal solution
under many circumstances. We also include a brief discussion and demonstration of the possibility of approximating a bolus+continuous dosing schedule by an intermittent schedule consisting of a series of bolus doses.

While the results obtained here provide what we 
believe to be an essential theoretical basis 
for thinking about the optimization of antimicrobial treatments, 
there are various complicating issues that should be taken into account in considering the application of these results in concrete settings. In section \ref{discussion} we address some of the limitations of the 
standard modelling framework employed in this work, and suggest directions for further investigation.

\section{The pharmacodynamic model}
\label{themodel}

In this section we describe the
modelling framwork which will be employed to 
study antimicrobial treatment schedules, which is standard 
in the field of pharmacodynamics (\cite{austin,bhagunde,bouvier,corvaisier,goranova,hoyle,kesisoglou, mouton,nielsen,nikolaou1,nikolaou2}).
The notation to be used is summarized in Table 1.

An antimicrobial treatment schedule will determine a function $C(t)$ ($t\geq 0$) describing the concentration 
of antimicrobial at the infection site as a function of time $t$, which we will call the {\it{concentration profile}}. We allow 
$C(t)$ to be an arbitrary non-negative function in the class  $L^1[0,\infty)$ of integrable functions.

The 
area under the concentration curve
\begin{equation}\label{auc}AUC=AUC[C(t)]=\int_0^\infty C(t)dt,\end{equation}
is a standard measure of the intensity of 
the antimicrobial treatment. Indeed it may be seen that the $AUC$ is proportional to 
the cumulative dosage of the antimicrobial
supplied, see equation \eqref{aucd}.

Denoting by $B(t)$ the size of the microbial population at time $t$, we use the standard constant-rate model 
of microbial growth in the absence of treatment 
$$\frac{dB}{dt}=rB,$$
where $r$ is the difference of the replication rate and the natural death rate, leading to exponential growth, with doubling time
\begin{equation}\label{dtime}T_2=\frac{\ln(2)}{r}.
\end{equation}
The antimicrobial effect is modelled using a function $k(c)$, known as the pharmacodynamic function (\cite{regoes}), or the kill curve (\cite{mueller}), which describes the kill-rate of the antimicrobial agent at concentration $c$. Note that we will be using lower-case $c$ to denote an arbitrary value of antimicrobial concentration, while the upper-case $C(t)$ is used for  describing the antimicrobial concentration as a function of time.

In the presence of antimicrobial, the microbial population is thus described by
\begin{equation}\label{model}\frac{dB}{dt}=[r-k(C(t))]B,\end{equation}
with solution
\begin{equation}\label{sol}B(t)=B_0\cdot e^{\int_0^t [r-k(C(s))]ds},\end{equation}
where $B_0$ is the initial microbial population size at time $t=0$.

The kill-rate function $k(c)$ will be assumed to have the following properties
\begin{itemize}
	\item[(A1)] $k(0)=0$, and $k(c)$ is continuous and monotone increasing on $[0,\infty)$, and twice differentiable for $c>0$.
	
	\item[(A2)] The kill rate 
	saturates at high 
	concentrations:
	\begin{equation}\label{maxkill}\lim_{c\rightarrow\infty} k(c)=k_{max}<\infty.\end{equation}
	
	\item[(A3)] The function 
	$k(c)$ satisfies one of the 
	following two conditions:
	
	(i) $k(c)$ is strictly {\bf{concave}}, that is $k''(c)<0$ for all $c\geq 0$.
	
	or
	
	(ii) $k(c)$ is
	{\bf{sigmoidal}}, that is, there exists a value $c_{infl}>0$ (the inflection point) so that 
	$$0<c<c_{infl} \;\;\Rightarrow\;\; k''(c)>0\;\;\;{\mbox{and}}\;\;\;c>c_{infl}\;\;\Rightarrow\;\; k''(c)<0.$$	
\end{itemize}

\begin{example}
	The most common functional form used for the pharmacodynamic function is the
	Hill function, also known as the 
	Sigmoid $E_{max}$ model (\cite{meibohm}) or the Zhi model (\cite{corvaisier,zhi}) 
	\begin{equation}\label{hill}k_H(c)=k_{max}\cdot\frac{c^{\gamma}}{C_{50}^{\gamma}+c^{\gamma}},\end{equation}
	where $\gamma>0$ is called the Hill exponent and
	$C_{50}$ is the half-saturation constant, the
	concentration at which the kill rate is half 
	of the maximal value $k_{max}$.
	In the survey of \cite{czock} one may find 
	tables with estimates of the parameters $C_{50},k_{max},\gamma$ for various combinations 
	of antimicrobials and microbial species, obtained through many empirical studies.
	
	The function $k_H(c)$ is concave 
	if $\gamma\leq 1$, and sigmoidal if $\gamma>1$, in which case the inflection point is given by
	$$c_{infl}=C_{50}\cdot \left(\frac{\gamma-1}{\gamma+1} \right)^{\frac{1}{\gamma}}.$$
	Our main results do not depend on this 
	specific functional form, but we will 
	apply the general results to this specific example, and obtain useful
	explicit expressions - see in particular sections \ref{hilf} and \ref{bchill}.
\end{example}

It will be useful to introduce the dimensionless parameter 
\begin{equation}\label{dalpha}
	\alpha=\frac{k_{max}}{r},
\end{equation}
measuring the maximal kill-rate of the 
antimicrobial relative to the natural 
microbial growth rate, which we will therefore
call the {\it{potency}} of the antimicrobial with 
respect to a microbial species.
We will make the standing assumption that
$\alpha>1$, 
that is $k_{max}>r$, which means that a sufficiently large concentration will lead to a negative net growth rate of the microbial population - if this is 
not the case then the antimicrobial is not effective. Under this assumption, and in view of the assumption (A1) above, there exists a unique value of $c$, denoted by $zMIC$, the pharmacodynamic minimal inhibitory concentration (\cite{bouvier,corvaisier}), also referred to as the stationary concentration (SC) (\cite{mouton,czock}), such that
\begin{equation}\label{dzMIC}k(zMIC)=r.\end{equation}

\begin{example}
	In the case that the kill-rate is given by 
	a Hill function \eqref{hill}, solving \eqref{dzMIC} for $zMIC$ and using \eqref{dalpha} gives
	\begin{equation}\label{zMIC}zMIC=C_{50}\cdot \left(\alpha-1\right)^{-\frac{1}{\gamma}}.\end{equation}
\end{example}

Note that within the framework of model \eqref{model} the microbial population cannot be reduced to $0$, since $B(t)$
given by \eqref{sol} is always positive, and indeed if $AUC$ is finite then
for large time the microbial population will recover, with $B(t)\rightarrow +\infty$ as $t\rightarrow \infty$. However,
in practice, reaching a sufficiently low value of $B(t)$ at some point in time, {\it{e.g.}} corresponding to less than one organism, implies eradication. 
The appropriate measure for the success of treatment is therefore the {\it{maximal}} reduction in the size of the microbial population achieved at some time. The reduction is standardly expressed on a logarithmic scale, by 
defining the log-reduction at time $T$
$$LR(T)=\log_{10}\left(\frac{B_0}{B(T)} \right)=\frac{1}{\ln(10)}\cdot \ln \left(\frac{B_0}{B(T)} \right),$$
(we use base $10$ in order to be consistent with the literature)
which in view of \eqref{sol} is given by 
\begin{equation}\label{LR}LR(T)=\frac{1}{\ln(10)}\cdot \int_0^T [k(C(t))-r]dt.\end{equation}
The {\it{maximal}} reduction afforded by
the concentration profile $C(t)$ is then
\begin{equation}\label{lrs}LR_{max}=LR_{max}[C(t)]=\max_{T\geq 0} LR(T).\end{equation}
We note that, since $LR(0)=0$ and $\lim_{T\rightarrow\infty}LR(T)=-\infty$, the maximum in \eqref{lrs} certainly exists.

Eradication of the infection thus 
corresponds to achieving $LR_{max}[C(t)]\geq LR_{target}$, where 
the value of $LR_{target}$ is given.

\begin{longtable}{p{0.2\textwidth}p{0.8\textwidth}}
	\label{notation}
	%	\begin{tabular}
		Symbol & Description\\
		\hline
		$B(t)$ & Microbial population size at time $t$. \\
		$B_0$ & Initial microbial population size.\\
		$C(t)$ & Antimicrobial concentration at time $t$. \\
		$AUC[C(t)]$ & Area under the time-concentration curve $C(t)$, see \eqref{auc}.\\
		$r$ & Microbial growth rate in absence of antimicrobial. \\
		$T_2$ & Microbial doubling time in absence of antimicrobial, see \eqref{dtime}.\\
		$k(c)$ & Kill-rate of antimicrobial, in dependence on its concentration.\\
		$k_{max}$ & Maximal kill-rate of antimicrobial, see  \eqref{maxkill}.\\
		$k_H(c)$ & Hill model for kill-rate, see  \eqref{hill}.\\
		$\gamma$ & Hill exponent, see \eqref{hill}.\\
		$C_{50}$ & Half saturation constant for Hill model, see \eqref{hill}.\\
		$\alpha$ & Antimicrobial potency relative to a microbial species, see \eqref{dalpha}.\\
		$c_{infl}$ & Inflection point of $k(c)$, in the sigmoidal case.\\
		$zMIC$ & Pharmacodynamic minimal inhibitory concentration, see  \eqref{dzMIC}.\\
		$LR(T)$ & log (base $10$) reduction of microbial population at time $T$, see \eqref{LR}.\\
		$LR_{max}[C(t)]$ & Maximal log reduction of microbial population corresponding to a given concentration profile, see \eqref{lrs}.\\
		$LR_{target}$ & Target log reduction of microbial 
		population for achieving eradication.\\
		$C_{opt}(t)$ & Optimal antimicrobial concentration profile, see Theorem \ref{mainr}.\\
		$c_{opt}$ & Optimal antimicrobial concentration, see Theorem \ref{mainr}.\\
		$T_{opt}$ & Duration of optimal concentration profile, see \eqref{topte}.\\
		$AUC_{opt}$ & Minimal $AUC$ attainable by any antimicrobial concentration profile achieving eradication, given by \eqref{aucopt}.\\
		$LR_{opt}$ & Maximal log reduction attainable by a concentration profile with given $AUC$, given by \eqref{lropt}.\\
		$T_{min}$ & Greatest lower bound for the time to achieve eradication, see \eqref{tmin}.\\
		$\mu$ & Antimicrobial decay rate.\\
		$\tau$ & Antimicrobial half-life, see \eqref{dtau}.\\
		$d(t)$ & Antimicrobial dosing rate at time $t$.\\
		$d_{bc}(t)$ & bolus+continuous dosing schedule, see \eqref{ds}.\\
		$d_{bc,opt}(t)$ & optimal bolus+continuous dosing schedule, see \eqref{ds1}.\\
		$T_{bc}$ & Duration of dosing for a bolus+continuous schedule.\\
		$T_{bc,opt}$ & Duration of the optimal dosing of bolus+continuous type, see \eqref{bart}.\\
		\hline\\
		%	\end{tabular}
	\caption{Notation}\label{tab:parameters}
\end{longtable}

\section{Optimizing treatment: the `ideal' concentration profile}
\label{optall}

Our aim is to choose a concentration profile $C(t)$, among {\it{all}} non-negative integrable functions on $[0,\infty)$, so as to minimize the $AUC$, while achieving a specified  log-reduction of the microbial load. We therefore formulate:

\begin{prob}\label{prmain}
	Given a target value $LR_{target}$, find a concentration profile $C(t)$ achieving $LR_{max}[C(t)]=LR_{target}$, where $LR_{max}$ is given by \eqref{lrs}, with
	the corresponding 
	$AUC$ (given by \eqref{auc}) as {\it{small}} as possible.
\end{prob}

The following theorem provides a complete 
solution to this problem.

\begin{theorem}\label{mainr}
	Assume (A1)-(A3) and $\alpha>1$. Let a value $LR_{target}$ be given. The unique solution of Problem \ref{prmain} is given by the 
	concentration profile 
	\begin{equation}\label{copt1}C_{opt}(t)=\begin{cases}
			c_{opt} & 0\leq t\leq T_{opt}\\
			0 & t>T_{opt}
		\end{cases},\;\;\;\;\end{equation}
	where $c_{opt}$ is the unique maximizer of
	the function $f_r:(0,\infty)\rightarrow \Real$ defined by
	\begin{equation}\label{deff}f_r(c)=\frac{k(c)-r}{c},\end{equation}
	which is also the unique solution of the equation
	\begin{equation}\label{copte}k'(c)=\frac{k(c)-r}{c},\end{equation}
	and 
	\begin{equation}\label{topte}T_{opt}= \frac{\ln(10)\cdot LR_{target}}{k(c_{opt})-r}.\end{equation}
	The solution $C_{opt}(t)$ leads to  eradication at time $T_{opt}$ (that is, $LR(T_{opt})=LR_{target}$), and 
	achieves the minimal value possible of $AUC$, given by
	\begin{equation}\label{aucopt}AUC_{opt}=T_{opt}\cdot c_{opt}=\ln(10)\cdot LR_{target}\cdot \frac{c_{opt}}{k(c_{opt})-r}.\end{equation}
\end{theorem}

Several notable consequences emerge from Theorem \ref{mainr}:

\begin{itemize}
	\item[(i)] The `ideal' dosing strategy is to keep 
	the concentration of antimicrobial {\it{constant}} at $c_{opt}$ for the duration of time determined by
	$T_{opt}$. 
	
	\item[(ii)] By \eqref{copte}, $c_{opt}$ depends only on the microbial growth rate $r$ and on the 
	kill-rate function $k(c)$, and not on the target reduction $LR_{target}$. Thus the given target reduction affects the optimal schedule only through its effect on the duration $T_{opt}$, which, by \eqref{topte}, depends linearly on $LR_{target}$.
	
	\item[(iii)] By \eqref{aucopt}, the minimal achievable $AUC$ depends linearly on the target reduction
	$LR_{target}$.
\end{itemize}

We note that equation \eqref{copte} leads to a graphical construction for 
obtaining the optimal concentration $c_{opt}$ - see Figure \ref{graphic}.

\begin{figure}
	\begin{center}
		\includegraphics[width=0.8\linewidth]{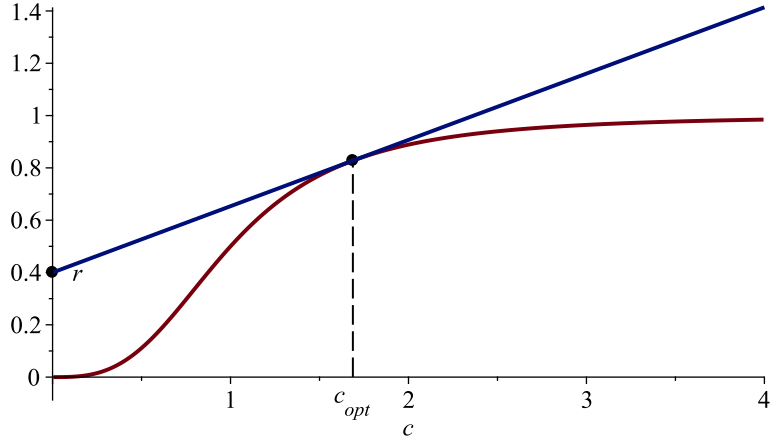}
	\end{center}
	\caption{Graphical construction for determining the optimal concentration $c_{opt}$. Plot the straight line which passes through the point $(0,r)$ and is tangent to the curve $k(c)$. The abscissa at the tangency point is the value $c_{opt}$. In this plot we use the Hill curve \eqref{hill} with $\gamma=3,C_{50}=1,k_{max}=1$, and $r=0.4$.}
	\label{graphic}
\end{figure}
In section \ref{hilf} the results of Theorem \ref{mainr} will be 
applied to the standard
Hill-type form of the killing curve, leading to explicit expressions for 
$c_{opt}$, $T_{opt}$.

We have formulated our problem and solution in 
terms of minimizing the $AUC$ subject to achieving a given log-reduction $LR_{target}$, but an equivalent problem is to maximize the log-reduction $LR_{max}$ 
subject to a given $AUC$ - since $AUC$ is proportional to the cumulative dosage this problem will arise if we want to maximize the efficacy of a given total dose of the antimicrobial. We can 
re-formulate the result of Theorem \ref{mainr} as follows:
\begin{corol}
	Given a value $AUC$, the concentration profile $C(t)$ satisfying
	$AUC[C(t)]=AUC$ and 
	inducing the maximal log-reduction $LR_{max}[C(t)]$ is given by \eqref{copt1}, where $c_{opt}$ is the 
	unique solution of \eqref{copte}, and 
	$$T_{opt}=\frac{AUC}{c_{opt}}.$$
	The value of $LR_{max}$ attained by this optimal profile is
	\begin{equation}\label{lropt}LR_{opt}=\frac{1}{\ln(10)}\cdot \frac{k(c_{opt})-r}{c_{opt}}\cdot AUC.
	\end{equation}
\end{corol}

We now begin the analysis leading to the proof of Theorem \ref{mainr}.
As a first step, we consider only the specific class of concentration profiles which 
take a constant value $c$ for a 
duration $T_{f}$ (the final time), that is:
\begin{equation}C(t)=\label{cp}\begin{cases}
		c & 0\leq t\leq T_f\\
		0 & t>T_f
\end{cases}.\end{equation}
Among these profiles, we will now find the one which achieves
the log-reduction $LR_{target}$ with minimal $AUC$. We emphasize that the restriction to profiles of the form \eqref{cp} is only temporary -  we will later prove that the resulting concentration
profile is in fact optimal among
{\it{all}} concentration profiles
satisfying $LR_{max}[C(t)]=LR_{target}$.

The log-reduction up to an arbitrary time $T\geq 0$,
corresponding to the concentration profile \eqref{cp}, is (see \eqref{LR}) $$LR(T)=\frac{1}{\ln(10)}\cdot \int_0^T [k(C(t))-r]dt=\frac{1}{\ln(10)}\cdot\begin{cases}
	T\cdot (k(c)-r) & T\leq T_{f}\\
	T_{f}\cdot k(c)-T\cdot r & T>T_{f}
\end{cases},$$
which is increasing for $T\leq T_{f}$ and decreasing for $T\geq T_{f}$,
hence
$$LR_{max}=\max_{T\geq 0}LR(T)=LR(T_{f})=\frac{1}{\ln(10)}\cdot T_{f}\cdot (k(c)-r).$$
Therefore to achieve a target 
log-reduction $LR_{target}$ using the constant 
concentration $c$ we need to choose
\begin{equation}\label{tf}
	T_f=\frac{\ln(10)\cdot LR_{target}}{k(c)-r}.\end{equation}
The $AUC$ corresponding to this profile is 
$$AUC=T_f\cdot c=\ln(10)\cdot LR_{target}\cdot \frac{c}{k(c)-r}.$$
Thus to minimize the $AUC$ we need to 
maximize the function $f_r:(0,\infty)\rightarrow \Real$ defined by \eqref{deff}
over all $c>0$. The existence of this maximizer, and 
the fact that it is the unique critical point of $f_r$, is shown in the following Lemma:
\begin{lemma}\label{fmax}
	Assume $k(c)$ satisfies (A1)-(A3), and $\alpha>1$.
	Then the function $f_r$ defined by \eqref{deff}
	has a unique critical point $c_{opt}$,
	that is a value satisfying
	\begin{equation}\label{copt0}f_r'(c)=0\;\;\Leftrightarrow\;\; k'(c)=\frac{k(c)-r}{c},\end{equation}
	which is a global maximizer of $f_r$.
	
	In the sigmoidal case (see  (A3)(ii)), we always have $c_{opt}>c_{infl}$, where $c_{infl}$ is the inflection point of $k(c)$.
\end{lemma}

\begin{proof}[Proof of Lemma \ref{fmax}]
	Using (A1),(A2) and the assumption $
	\alpha>1$, we have
	$$c<zMIC\;\;\Rightarrow\;\; f_r(c)<0,\;\;\;c>zMIC\;\;\Rightarrow\;\; f_r(c)>0,$$
	$$\lim_{c\rightarrow+\infty}f_r(c)=0.$$
	These facts imply that 
	$f_r(c)$ has a global maximizer 
	on $(0,\infty)$, which we denote by
	$c_{opt}$. 
	
	It remains to show that
	$c_{opt}$ is the unique critical point of $f_r(c)$.
	We have
	\begin{equation}\label{df}f_r'(c)=\frac{k'(c)-\frac{k(c)}{c}}{c}+\frac{r}{c^2}
		=\frac{r-h(c)}{c^2}
	\end{equation}
	where 
	\begin{equation}\label{defh}h(c)=k(c)-ck'(c),\end{equation}
	so any critical point of $f_r(c)$ 
	satisfies $h(c)=r$.
	Note that 
	\begin{equation}\label{hd}h'(c)=-ck''(c),\end{equation}
	hence:
	
	(a) if $k$ is concave then 
	$h'(c)>0$, so $h(c)$ is increasing in $(0,\infty)$, hence the critical point $c_{opt}$ of $f_r(c)$ is unique. 
	
	(b) If 
	$k$ is sigmoidal, then $h'(c)>0$ for $c>c_{infl}$, hence $h(c)$ is 
	increasing in this range, so that
	$f_r(c)$ has at most one critical point 
	in the interval $[c_{infl},\infty)$.
	To prove uniqueness it therefore suffices to show that $f_r(c)$ has no
	critical point in $[0,c_{infl}]$.
	But note that since
	$k(c)$ is convex on 
	$[0,c_{infl}]$, hence $k'(c)$ is increasing on this interval, we have
	\begin{equation}\label{ii}c\in (0,c_{infl}]\;\;\Rightarrow\;\;k'(c)=\frac{1}{c}\int_0^c k'(c)du>\frac{1}{c}\int_0^c k'(u)du=\frac{k(c)}{c},\end{equation}
	hence, by \eqref{df}, $f_r'(c)>0$ for
	$c\in (0,c_{infl}]$. Note that this also shows that $c_{opt}>c_{infl}$.
\end{proof}

The above considerations show that the optimal
concentration profile, among those of the form 
\eqref{cp} which achieve $LR_{target}$, is given by \eqref{copt1}, with 
$c_{opt}$ given by \eqref{copte}, and $T_{opt}$ given by \eqref{topte}.

We now show that the profile
$C_{opt}(t)$ is in fact optimal among
{\it{all}} concentration profiles which achieve
$LR_{target}$, and is thus the solution to 
Problem \ref{prmain}:

\begin{proof}[Proof of Theorem \ref{mainr}]
	We have, for any 
	$c>0$,
	$$k(c)-r=c\cdot \frac{k(c)-r}{c}\leq c \cdot \max_{c'>0}\frac{k(c')-r}{c'},$$
	and in view of Lemma \ref{fmax} this implies
	\begin{equation}\label{bin}
		k(c)-r\leq c\cdot \frac{k(c_{opt})-r}{c_{opt}},
	\end{equation}
	which is also valid for $c=0$.
	
	Let $C(t)$ be {\it{any}} concentration profile with
	$LR_{max}[C(t)]=LR_{target}$, and let $T^*>0$
	be a value for which $LR(T^*)=LR_{max}[C(t)]$. We then have, using
	(\ref{bin})
	\begin{eqnarray}\label{ki1}\ln(10)\cdot LR_{target}&=&\ln(10)\cdot LR(T^*)= \int_0^{T^*} [k(C(t))-r]dt\\&\leq& \frac{k(c_{opt})-r}{c_{opt}}\cdot \int_0^{T^*} C(t)dt\nonumber\\&\leq&  \frac{k(c_{opt})-r}{c_{opt}}\cdot \int_0^\infty C(t)dt= \frac{k(c_{opt})-r}{c_{opt}}\cdot AUC\nonumber\end{eqnarray}
	so that
	\begin{equation}\label{ki2}AUC\geq\ln(10)\cdot LR_{target}\cdot\frac{c_{opt}}{k(c_{opt})-r}= AUC_{opt},\end{equation}
	where $AUC_{opt}$ is given by \eqref{aucopt}, 
	so we see that $AUC$ cannot be made 
	smaller $AUC_{opt}$. Therefore $C_{opt}(t)$ is a minimizer. 
	
	To show uniqueness of the 
	minimizer, note that if we have
	equality in \eqref{ki2}, hence in 
	\eqref{ki1}, then it must be the case that, for almost every $t\in [0,T^*]$,
	$$k(C(t))-r=\frac{k(c_{opt})-r}{c_{opt}}\cdot C(t)\;\;\Rightarrow\;\;\frac{k(C(t))-r}{C(t)}=\max_{c> 0}\frac{k(c)-r}{c},$$
	implying that $C(t)=c_{opt}$, for almost every $t\in [0,T^*]$, as well as that
	$$\int_0^{T^*} C(t)dt=\int_0^\infty C(t)dt,$$
	implying that $C(t)=0$ for {\it{a.e.}}
	$t\geq T^*$. We therefore have
	$$LR_{max}[C(t)]=\frac{1}{\ln(10)}\cdot T^*\cdot  (k(c_{opt})-r),$$ implying that
	$$T^*=\frac{\ln(10)\cdot LR_{target}}{k(c_{opt})-r}=T_{opt}.$$
	We have thus shown that $C(t)=C_{opt}(t)$ for {\it{a.e.}} $t$, establishing uniqueness.
\end{proof}
The assumption (A3) that $k(c)$ is either concave or sigmoidal, plays only a limited role in the proof of Theorem \ref{mainr} - it was used to prove that the maximizer $c_{opt}$ of $f_r(c)$ is unique, and that $c_{opt}$ is the unique solution of \eqref{copte} (Lemma \ref{fmax}). The {\it{existence}} of a global maximizer 
of $f_r(c)$ follows from (A1),(A2) (see proof of Lemma 1), and any such global maximizer $c_{opt}$ gives rise to a solution $C_{opt}(t)$ of Problem \ref{prmain}, defined by \eqref{copt1}.
If (A3) does not hold (note that (A1),(A2) imply that $k(c)$ cannot be convex, but it can have more than one inflection point), then it is possible that $f_r(c)$ will have critical points which are not global maximizers, which will thus not give rise to solutions of Problem \ref{prmain}, or that the global maximum will be attained at more than one point, in which case there will be two or more solutions of Problem \ref{prmain}, all giving rise to the same $AUC$. However, kill-rate functions employed in practice are either concave or sigmoidal, hence (A3) is satisfied.

We conclude this section by considering a variant of 
Problem \ref{mainr}, which will be relevant under some circumstances.
A notable feature of the the optimization problem that we have posed and solved above is the fact that neither the duration of treatment, nor the time $T^*$ at which eradication is achieved, are constrained in advance - rather they are determined as part of the solution of the problem (and turn out to be equal). This allows for maximal flexibility in achieving the aim of minimizing the $AUC$. A potential disadvantage arises in case that the duration $T_{opt}$ until eradication is judged to be too long, in which case we might be willing to accept a higher value of $AUC$ in order achieve a lower value of the time $T^*$ to eradication. We can thus formulate a {\it{time-restricted}} version of Problem \ref{prmain}, in which we pre-determine a maximal allowed duration $\hat{T}$ to eradication:
\begin{prob}[Time-restricted version of Problem \ref{prmain}]\label{prmodified}
	Given a target value $LR_{target}$ and a maximal duration $\hat{T}$, find a concentration profile $C(t)$ such that there exists a time $T\leq \hat{T}$
	for which $LR(T)=LR_{target}$ (where $LR(T)$ is given by \eqref{LR}), with
the corresponding 
$AUC$ (given by \eqref{auc}) as {\it{small}} as possible.
\end{prob}

The solution of this modified problem, given in the following theorem, is different for $\hat{T}$ in different ranges.

\begin{theorem}\label{modified}
Assume (A1)-(A3) and $\alpha>1$. Let a value $LR_{target}$ and a maximal duration
$\hat{T}$ be given. Then

\begin{itemize}
\item[(i)] If $\hat{T}\geq T_{opt}$, the solution of Problem \ref{prmodified} coincides with the solution $C_{opt}(t)$ of problem \ref{prmain}.

\item[(ii)] If 
\begin{equation}\label{tmin}
\hat{T}\leq T_{min}\doteq\frac{\ln(10)\cdot LR_{target}}{k_{max}-r}
\end{equation}
then it is {\bf{impossible}} to achieve eradication within such a short time, so that the problem has no solution.

\item[(iii)] In the intermediate range,
$T_{min}<\hat{T}<T_{opt}$, the solution of Problem \ref{prmodified}
consists of a constant concentration
$\hat{c}$, applied for the entire time
duration up to time $\hat{T}$:
\begin{equation}\label{chat}\hat{C}(t)=\begin{cases}
	\hat{c} & 0\leq t\leq \hat{T}\\
	0 & t>\hat{T},
\end{cases},\end{equation}
where $\hat{c}$ is determined by the 
equation
\begin{equation}\label{erad}
\hat{T}\cdot [k(\hat{c})-r]=\ln(10)\cdot LR_{target}.
\end{equation}
\end{itemize}
\end{theorem}

Thus, while the lowest $AUC$ will be achieved by the 
profile $C_{opt}(t)$ given by \eqref{copt1} in time $T_{opt}$, if one wishes to achieve eradication 
in a shorter time $\hat{T}<T_{opt}$, one can do so, at the price of a higher $AUC$, as long as $\hat{T}>T_{min}$. In this case the lowest $AUC$ that can be achieved - which will be higher than $AUC_{opt}$ - will be attained by chosing the profile 
$\hat{C}(t)$ according to \eqref{chat},\eqref{erad}.

We note that the numerical exmaples in the following section show that, for realistic parameter values, the optimal duration $T_{opt}$ is quite reasonable - on the order of a few days when the 
doubling time of the microbial population in the absence of antimicrobial is on the order of hours. Hence in most practical cases there might not be need to restrict the time duration to be below 
$T_{opt}$, though this might change when considering critical cases in which fast reduction of the microbial load is required.

\begin{proof}[Proof of Theorem \ref{modified}]
(i) is trivially true, since if $\hat{T}>T_{opt}$ then $C_{opt}(t)$ leads to  eradication before time 
$\hat{T}$, and, by Theorem \ref{mainr} it achieves the lowest possible value of $AUC$.

(ii) If $\hat{T}\leq T_{min}$ then for any
$T\leq \hat{T}$ we have
$$\ln(10)\cdot LR(T)= \int_0^T [k(C(t))-r]dt< T\cdot [k_{max}-r]\leq T_{min}\cdot [k_{max}-r]=\ln(10)\cdot LR_{target},$$
so that the target reduction cannot be achieved.

(iii) Assume now that \begin{equation}\label{inter}T_{min}<\hat{T}<T_{opt}.
\end{equation}
Note first that $\hat{T}>T_{min}$ implies
that, using (A2)
$$\lim_{c\rightarrow \infty}\hat{T}\cdot [k(c)-r]=\hat{T}\cdot [k_{max}-r]>T_{min}\cdot [k_{max}-r]=\ln(10)\cdot LR_{target},$$
hence by (A1) the equation \eqref{erad} has
a unique solution $\hat{c}$.

\eqref{erad} implies that
$$\frac{1}{\ln(10)}\int_0^{\hat{T}}[k(\hat{C}(t))-r]dt=\hat{T}\cdot [k(\hat{c})-r]=LR_{target},$$
so that the profile $\hat{C}(t)$ achieves 
eradication at time $\hat{T}$.

To show that $\hat{C}(t)$, defined by \eqref{chat} is the solution of Problem \ref{prmodified}, we assume that $C(t)$ is any concentration profile for which $LR(T)=LR_{target}$, where $T\leq \hat{T}$, and we need to show that if $AUC[C(t)]\geq AUC[\hat{C}(t)]$, with equality iff $C(t)=\hat{C}(t)$ for {\it{a.e.}} $t$. 

By \eqref{erad},\eqref{inter}, and 
\eqref{topte}, we have 
$$k(\hat{c})-r=\frac{\ln(10)\cdot LR_{target}}{\hat{T}}> \frac{\ln(10)\cdot LR_{target}}{T_{opt}}=k(c_{opt})-r,$$
which, by (A1), implies 
\begin{equation}\label{hgo}\hat{c}>c_{opt}.\end{equation}
We now define 
\begin{equation}\label{hr}\hat{r}=k(\hat{c})-\hat{c}k'(\hat{c}),\end{equation}
and claim that 
\begin{equation}\label{rr}\hat{r}>r.
\end{equation}
To show this we recall that 
in the proof of Lemma \ref{fmax} it was shown that the function $h(c)$ defined by
\eqref{defh} is monotone increasing for all $c>0$ in the concave case, and for all $c>c_{infl}$ in the sigmoidal case, in which case we also have
$c_{opt}>c_{infl}$. Therefore, since $h(c_{opt})=r$ and $h(\hat{c})=\hat{r}$, \eqref{hgo} implies \eqref{rr}.

By \eqref{hr} we have
$$k'(\hat{c})=\frac{k(\hat{c})-\hat{r}}{\hat{c}},$$
which, using Lemma \ref{fmax}, applied with $\hat{r}$ replacing $r$, implies that
$\hat{c}$ is the maximizer of $f_{\hat{r}}(c)$, so that
$$\frac{k(\hat{c})-\hat{r}}{\hat{c}}=\max_{c>0}\frac{k(c)-\hat{r}}{c}.$$
We therefore have, for all $c>0$,
\begin{equation}\label{ink}k(c)-\hat{r}\leq c\cdot \frac{k(\hat{c})-\hat{r}}{\hat{c}}.\end{equation}
If $C(t)$ is a concentration profile for which $LR(T)=LR_{target}$, where $T\leq \hat{T}$, then we have, using \eqref{rr},\eqref{ink}
\begin{eqnarray}\label{sin1}\ln(10)\cdot LR_{target}&=&\int_0^T [k(C(t))-r]dt=\int_0^T [k(C(t))-\hat{r}]dt
+T\cdot (\hat{r}-r)\nonumber\\
&\leq& \int_0^T [k(C(t))-\hat{r}]dt
+\hat{T}\cdot (\hat{r}-r)
\leq \frac{k(\hat{c})-\hat{r}}{\hat{c}} \int_0^T C(t)dt
+\hat{T}\cdot (\hat{r}-r)\nonumber\\
&\leq& \frac{k(\hat{c})-\hat{r}}{\hat{c}} \cdot AUC[C(t)]
+\hat{T}\cdot (\hat{r}-r),
\end{eqnarray}
which, togther with \eqref{erad}, implies
\begin{eqnarray*}\label{sin2}AUC[C(t)]&\geq& \left[\ln(10)\cdot LR_{target}-\hat{T}\cdot (\hat{r}-r)\right]\cdot \frac{\hat{c}}{k(\hat{c})-\hat{r}}\nonumber\\&=&\left[\hat{T}\cdot [k(\hat{c})-r]-\hat{T}\cdot (\hat{r}-r)\right]\cdot \frac{\hat{c}}{k(\hat{c})-\hat{r}}=\hat{T}\cdot \hat{c}=AUC[\hat{C}(t)],
\end{eqnarray*}
proving that $\hat{C}(t)$ indeed 
attains the minimal $AUC$ among the relevant concentration profiles. To show uniqueness, note that the equality $AUC[C(t)]=AUC[\hat{C}(t)]$
can hold only if all inequalities in \eqref{sin1} are in fact equalities, which in particular 
implies $T=\hat{T}$ and 
$k(C(t))=k(\hat{c})$, 
hence
$C(t)=\hat{c}$ for {\it{a.e.}} $0\leq t\leq \hat{T}$,  and also that 
$\int_{\hat{T}}^\infty C(t)dt=0$,
so that $C(t)=0$ for {\it{a.e.}} $t> \hat{T}$, hence $C(t)=\hat{C}(t)$ {\it{a.e.}}.	
\end{proof}

\section{Application to Hill-type pharmacodynamic  functions}
\label{hilf}

We now specialize the results to the case that the kill-rate function  is the 
Hill function $k_H(c)$ defined by \eqref{hill}, which 
allows us to obtain explicit expressions for 
the quantities of interest. We use 
these expressions to study the 
dependence of $c_{opt}$ and 
$T_{opt}$ on the relevant parameters.

In the case of a Hill-type pharmacodynamic function \eqref{deff} gives
$$f_r(c)=\frac{1}{c}\left(k_{max}\cdot\frac{c^{\gamma}}{C_{50}^{\gamma}+c^{\gamma}}-r\right)
=\frac{r}{c}\left(\alpha\cdot\frac{c^{\gamma}}{C_{50}^{\gamma}+c^{\gamma}}-1\right),$$
where $\alpha$ is the antimicrobial potency given by \eqref{dalpha}.
Solving the equation 
$f_r'(c)=0$, which is equivalent to
$$(\alpha-1)\left(\frac{c}{C_{50}}\right)^{2\gamma}-((\gamma-1)\alpha+2))\left(\frac{c}{C_{50}}\right)^\gamma -1=0,$$
we find that the optimal concentration is
\begin{equation}\label{copth}c_{opt}=c_{opt}(\gamma,\alpha)=C_{50}\cdot \left(\alpha-1\right)^{-\frac{1}{\gamma}}\cdot  \left[1+\alpha\cdot \left( \frac{\gamma-1}{2} + \sqrt{\left(\frac{\gamma-1}{2} \right)^2+\frac{\gamma}{\alpha} }\right)\right]^{\frac{1}{\gamma}},
\end{equation}
and substituting this into \eqref{topte} we find that the time
$T_{opt}$ for which this concentration
should be maintained is
\begin{equation}\label{topth}
	T_{opt}=T_{opt}(\gamma,\alpha)=T_2\cdot \frac{\ln(10)\cdot LR_{target}}{\ln(2)(\alpha-1)}\cdot 
	\left[1+\frac{\alpha}{\gamma}\cdot\left(\sqrt{\left(\frac{\gamma-1}{2}\right)^2+\frac{\gamma}{\alpha} }-\frac{\gamma-1}{2}\right)\right],
\end{equation}
where $T_2$ is the microbial doubling time in the absence of antimicrobial (see \eqref{dtime}).

Note that:
\begin{itemize}
	\item The optimal concentration $c_{opt}$ is 
	linearly dependent on the half-saturation
	concentration $C_{50}$. Therefore in our presentation of numerical results below we provide the values of the dimensionless ratio $\frac{C_{opt}}{C_{50}}$.
	
	\item The optimal duration
	$T_{opt}$ is linearly dependent 
	on the target log-reduction
	$LR_{target}$, as well as on the microbial doubling time 
	$T_2$. Therefore in presentation of numerical results we provide the values of the dimensionless ratio  $\frac{T_{opt}}{T_2}$.
	Note that
	$T_{opt}$ does {\it{not}} depend on the value of the 
	half-saturation constant $C_{50}$.
	
	\item As a consequence of the above and of \eqref{aucopt},
	the value $AUC_{opt}$ depends linearly on both 
	$T_2$ and $C_{50}$, so that in the presentation of numerical
	results we provide the value of the dimensionless ratio
	$\frac{AUC}{T_2\cdot C_{50}}$.
	
	\item In the special case $\gamma=1$, which is commonly employed (referred to as the $E_{max}$ model \cite{meibohm}), the somewhat complicated expressions given above simplify to
	$$\gamma =1\;\;\Rightarrow\;\;
	c_{opt}=\frac{C_{50}}{\sqrt{\alpha}-1},\;\;\;
		T_{opt}=T_2\cdot \frac{\ln(10)\cdot LR_{target}}{\ln(2)(
			\sqrt{\alpha}-1)}.
	$$
\end{itemize}

Table 2 presents the optimal concentration, duration, and $AUC$ for 
parameter values in the range which is typical for most antimicrobials. 
According to studies estimating parameters for 
various antimicrobials, reviewed in (\cite{czock}),
the Hill coefficient of $\gamma$ is in most cases 
in the range $0.5-5$, and the potency $\alpha=\frac{k_{max}}{r}$ is mostly in the range 
$2-6$. In the calculation of $T_{opt}$ and 
$AUC_{opt}$ we have taken $LR_{target}=7$ -- in view of the linear dependence of these quantities on $LR_{target}$, to obtain $T_{opt},AUC_{opt}$ for any
value of $LR_{target}$ one simply needs to multiply the value in the table by $\frac{LR_{target}}{7}$.
Note that the quantities $\frac{c_{opt}}{C_{50}}$,$\frac{T_{opt}}{T_2}$,$\frac{AUC_{opt}}{T_2\cdot C_{50}}$ presented in these tables, as well as 
	the parameters $\alpha,\gamma$ which they depend on, are all non-dimensional, so that the numbers are valid independently of the units used to measure concentration and time.

We can use the explicit expressions
\eqref{copth},\eqref{topth} to study the nature of the 
dependence of $c_{opt}$ and $T_{opt}$ on the  drug potency $\alpha$ and the Hill exponent $\gamma$. The results are 
given in Propositions \ref{proper1},\ref{proper2} - we 
omit the derivation of these results since they are routine applications of elementary calculus arguments.

The following proposition shows that higher antimicrobial potency $\alpha$ 
leads the optimal concentration profile to involve both a lower 
concentration and a shorter duration
(see also Figure \ref{res}).

\begin{prop}\label{proper1}
	(a) For fixed $\gamma>0$, the function
	$c_{opt}(\alpha)=c_{opt}(\gamma,\alpha)$
	($\alpha>1$) is monotone decreasing, with
	\begin{equation}\label{lx}\lim_{\alpha\rightarrow 1+}c_{opt}(\alpha)=+\infty.\;\;\end{equation}
	and
	\begin{itemize}
		\item[(i)] In the concave case $\gamma\leq 1$:
		$\lim_{\alpha\rightarrow \infty}c_{opt}(\alpha)=0$.
		
		\item[(ii)] In the sigmoidal case $\gamma> 1$:
		$\lim_{\alpha\rightarrow \infty}c_{opt}(\alpha)=C_{50}\cdot\left(\gamma-1\right)^{\frac{1}{\gamma}}>0$.
	\end{itemize}
	
	\noindent
	(b) For fixed $\gamma>0$, the function
	$T_{opt}(\alpha)=T_{opt}(\gamma,\alpha)$
	($\alpha>1$)
	is monotone decreasing, with
	\begin{equation}\label{llx}\lim_{\alpha\rightarrow 1+}T_{opt}(\alpha)=+\infty,\;\;\;\lim_{\alpha\rightarrow \infty}T_{opt}(\alpha)=0.\end{equation}
\end{prop}

\begin{table}
	\centering
	
	\begin{minipage}{.4\linewidth}
		\centering
		{\bfseries\strut Optimal concentration $\frac{c_{opt}}{C_{50}}$}
		\begin{tabular}{|l||*{5}{c|}}\hline
			\backslashbox{$\gamma$}{$\alpha$}
			&\makebox[2em]{2.0}&\makebox[2em]{3.0}&\makebox[2em]{4.0}
			&\makebox[2em]{5.0}&\makebox[2em]{6.0}\\\hline\hline
		0.5 & 2.62 & 0.71 & 0.33 & 0.19 & 0.13\\\hline
		1.0 & 2.41 & 1.37 & 1.00 & 0.81 & 0.69\\\hline
		2.0 & 2.06 & 1.64 & 1.47 & 1.37 & 1.31\\\hline
		3.0 & 1.83 & 1.60 & 1.51 & 1.46 & 1.42\\\hline
		4.0 & 1.69 & 1.54 & 1.48 & 1.44 & 1.42 \\\hline
		5.0 & 1.59 & 1.48 & 1.43 & 1.41 & 1.39 \\\hline
		\end{tabular}

	\end{minipage}%
	\quad % ----------------------------------
	\bigskip
	
	\begin{minipage}{.4\linewidth}
		\centering
		{\bfseries\strut Optimal duration $\frac{T_{opt}}{T_2}$}
		\begin{tabular}{|l||*{5}{c|}}\hline
			\backslashbox{$\gamma$}{$\alpha$}
			&\makebox[2em]{2.0}&\makebox[2em]{3.0}&\makebox[2em]{4.0}
			&\makebox[2em]{5.0}&\makebox[2em]{6.0}\\\hline\hline
			0.5 & 98.5 & 62.5 & 50.1 & 43.8 & 39.9 \\\hline
			1.0 & 56.1 & 31.8 & 23.3 & 18.8 & 16.0 \\\hline
			2.0 & 37.6 & 19.6 & 13.4 & 10.3 & 8.3 \\\hline
			3.0 & 32.3 & 16.4 & 11.1 & 8.4 & 6.7 \\\hline
			4.0 & 29.8 & 15.1 & 10.1 & 7.6 & 6.1 \\\hline
			5.0 & 28.4 & 14.3 & 9.6 & 7.2 & 5.8 \\\hline
		\end{tabular}
		
	\end{minipage}%
	\quad % ----------------------------------
	\bigskip
	
	\begin{minipage}{.4\linewidth}
		\centering
		{\bfseries\strut Optimal AUC $\frac{AUC_{opt}}{T_2\cdot C_{50}}$}
		\begin{tabular}{|l||*{5}{c|}}\hline
			\backslashbox{$\gamma$}{$\alpha$}
			&\makebox[2em]{2.0}&\makebox[2em]{3.0}&\makebox[2em]{4.0}
			&\makebox[2em]{5.0}&\makebox[2em]{6.0}\\\hline\hline
			0.5 & 257.9 & 44.4 & 16.7 & 8.5 & 5.1 \\\hline
			1.0 & 135.5 & 43.4 & 23.3 & 15.2 & 11.1 \\\hline
			2.0 & 77.4 & 32.1 & 19.7 & 14.1 & 10.9 \\\hline
			3.0 & 59.1 & 26.4 & 16.7 & 12.2 & 9.6 \\\hline
			4.0 & 50.3 & 23.1 & 14.9 & 11.0 & 8.7 \\\hline
			5.0 &  45.0 & 21.1 & 13.7 & 10.1 & 8.0 \\\hline
		\end{tabular}

	\end{minipage}%
	\quad % ----------------------------------
	\bigskip 
	
			\caption{Optimal concentration profiles for 
					different values of $\alpha=\frac{k_{max}}{r}$ and of the Hill exponent $\gamma$.
					Top: Optimal concentration $\frac{c_{opt}}{C_{50}}$, Middle: optimal duration $\frac{T_{opt}}{T_2}$ for achieving log-reduction $LR_{target}=7$.
					Bottom: value of $\frac{AUC}{T_2\cdot C_{50}}$ attained by the optimal schedule, assuming $LR_{target}=7$. }
	
\end{table}

\begin{figure}
	\begin{center}
		\includegraphics[width=0.45\linewidth]{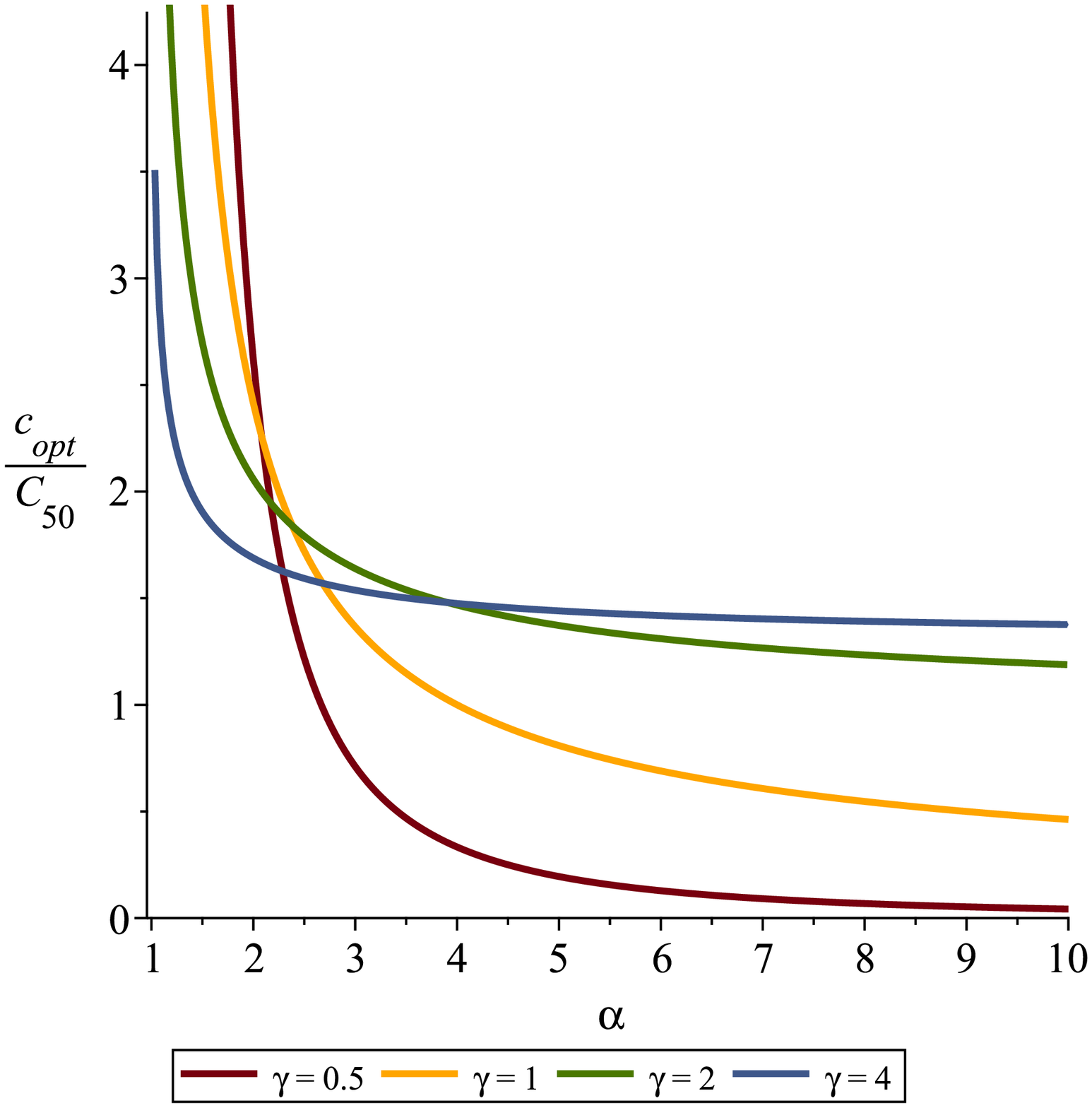}
		\includegraphics[width=0.45\linewidth]{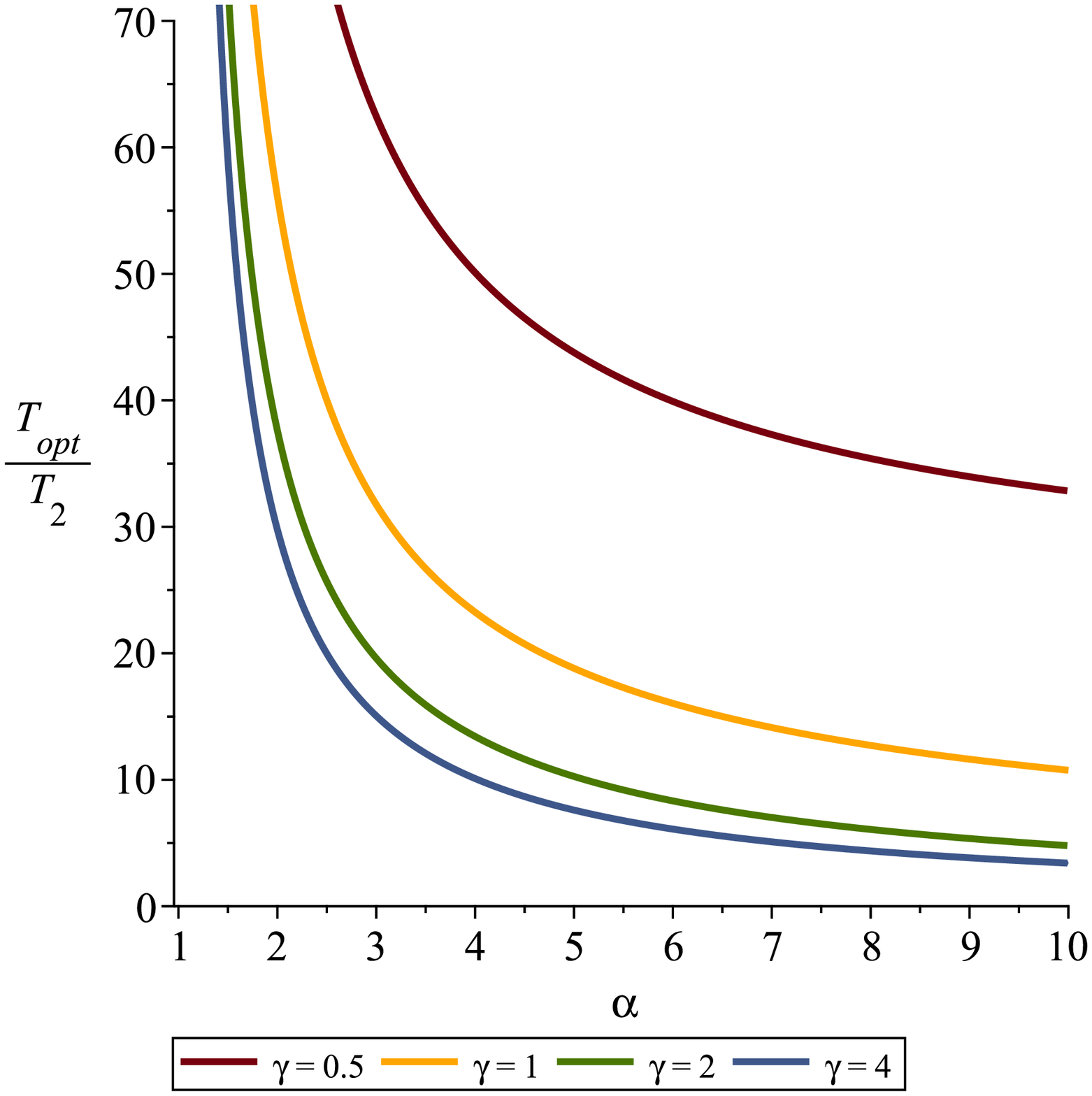}
	\end{center}
	\caption{The parameters defining the `ideal' concentration profile as functions of the drug potency $\alpha=\frac{k_{max}}{r}$, for 
		different values of the Hill exponent $\gamma$.
		Left: the ratio of $c_{opt}$ to the half-saturation value $C_{50}$ of the antimicrobial. Right: the ratio of $T_{opt}$ to the
		microbial doubling time $T_2$. Here it is assumed that $LR_{target}=7$.}
	\label{res}
\end{figure}

\begin{figure}
	\begin{center}
		\includegraphics[width=0.45\linewidth]{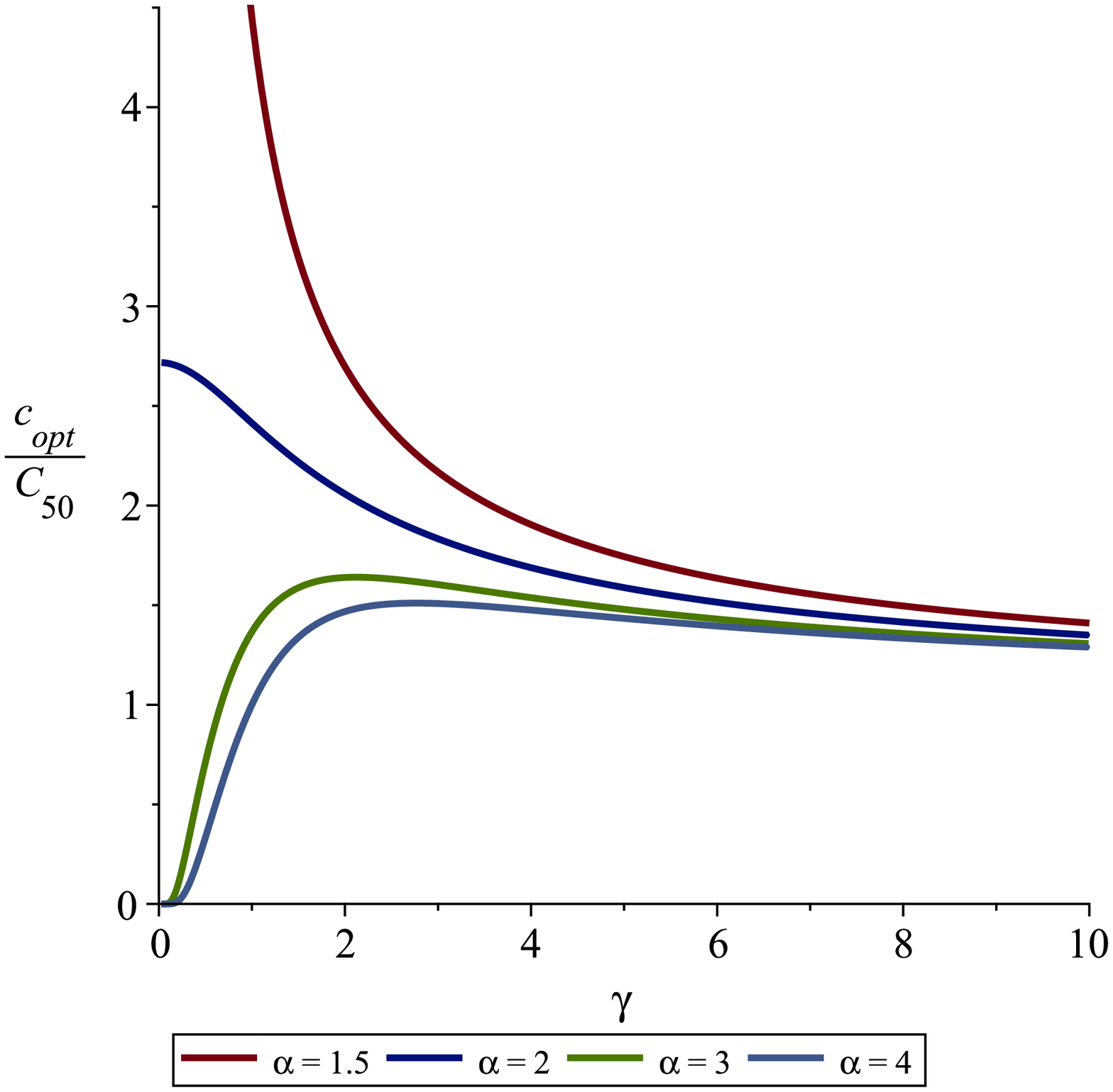}
		\includegraphics[width=0.45\linewidth]{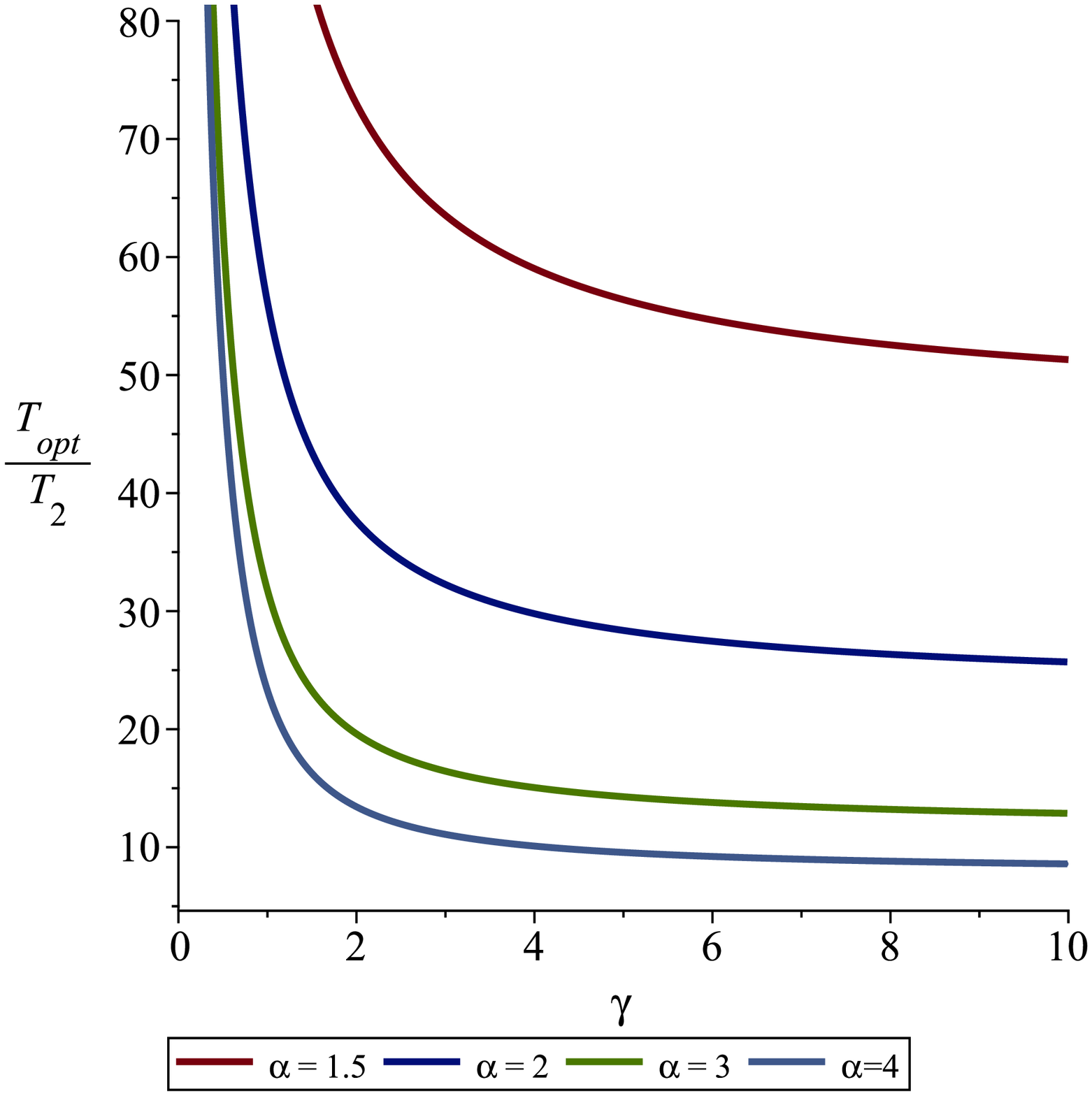}
	\end{center}
	\caption{The parameters defining the optimal concentration profile as functions of the Hill exponent $\gamma$, for 
		different values of the potency $\alpha$.
		Left: the ratio of $c_{opt}$ to the half-saturation value $C_{50}$ of the antimicrobial. Right: the ratio of $T_{opt}$ to the
		microbial doubling time $T_2$. Here it is assumed that $LR_{target}=7$.}
	\label{res2}
\end{figure}

The dependence 
of $c_{opt}$, $T_{opt}$ on the Hill exponent
$\gamma$ is described in the next proposition. Note that the shape of 
the function $c_{opt}(\gamma)$ is different for 
$\alpha<2$, $\alpha=2$, and $\alpha>2$ - see also Figure \ref{res2}.

\begin{prop}\label{proper2}
	(a) For fixed $\alpha>1$, the function
	$c_{opt}(\gamma)=c_{opt}(\gamma,\alpha)$ satisfies
	\begin{equation}\label{lllx}\lim_{\gamma\rightarrow \infty }c_{opt}(\gamma)=C_{50},\end{equation}
	and
	\begin{itemize}
		\item[(i)] If $1<\alpha<2$ then $c_{opt}(\gamma)$ is
		monotone decreasing, with 
		$\lim_{\gamma\rightarrow 0+}c_{opt}(\gamma)=+\infty$.
		
		\item[(ii)] If $\alpha=2$ then $c_{opt}(\gamma)$
		is monotone decreasing, with 
		$\lim_{\gamma\rightarrow 0+}c_{opt}(\gamma)=e\cdot C_{50}$.
		
		\item[(iii)] If $\alpha>2$, then $c_{opt}(\gamma)$ is increasing for small 
		$\gamma$ and decreasing for large $\gamma$, and 
		$\lim_{\gamma\rightarrow 0+}c_{opt}(\gamma)=0$.
	\end{itemize}
	
	\noindent
	(b) For any $\alpha>1$, the function $T_{opt}(\gamma)=T_{opt}(\gamma,\alpha)$ is monotone 
	decreasing, with
	$$\lim_{\gamma\rightarrow 0+}T_{opt}(\gamma)=+\infty,\;\;\lim_{\gamma\rightarrow \infty}T_{opt}(\gamma)=0.$$
\end{prop}

\section{Pharmacokinetic considerations: the optimal bolus+continuous dosing schedule}
\label{pharmacokinetic}

In the preceding analysis we considered arbitrary antimicrobial concentration profiles $C(t)$. In practice,  however, the concentration profile cannot be chosen at will, since it is the result of a dosage plan and of the pharmacokinetics of the drug. Thus,
while the concentration profile 
given by Theorem \ref{mainr} is the 
optimal one, we show below that it is impossible to achieve this profile precisely, due the fact that the optimal 
profile $C_{opt}(t)$ has a discontinuity at
$t=T_{opt}$, while realistic pharmacokinetics 
precludes such a sharp cutoff.
It then becomes of interest to approximate the 
`ideal' concentration profile, to the extent possible, by a {\it{pharmacokinetically feasible}} one. We will consider one simple and 
natural method of doing so, and determine its
optimal version.

We assume a basic one-compartment pharmacokinetic model with first order degradation
kinetics - a reasonable choice for most commonly prescribed antimicrobials (\cite{bouvier}).
The dosing rate - the rate at which 
antimicrobial is added into the compartment, will be denoted by $d(t)$.
The concentration profile is then given by
the solution of the differential equation
\begin{equation}\label{pk}\frac{dC}{dt}=V^{-1}d(t)-\mu C(t),\;\;\;C(0)=0,\end{equation}
where $V$ is the volume of distribution and $\mu$ is the degradation/removal rate of the antimicrobial, so that
\begin{equation}\label{dtau}\tau=\frac{\ln(2)}{\mu}\end{equation}
is the antimicrobial half-life. 
Making the natural assumption that the cumulative dose
$$D=\int_0^\infty d(t)dt$$
is finite, it follows that
$C(\infty)=\lim_{t\rightarrow \infty}C(t)=0$, so that, by integrating
\eqref{pk} over $[0,\infty)$ we obtain
$$0=C(\infty)-C(0)=V^{-1}\int_0^\infty d(t)dt-\mu\cdot AUC\;\;$$
hence
\begin{equation}\label{aucd}AUC=\mu^{-1}V^{-1}D,\end{equation}
a relation that is well-known in pharmacokinetics (\cite{derendorf,rescigno}).
Note that this shows that the objective of minimizing the $AUC$ is equivalent to that of 
minimizing the cumulative dose. An analogous linear relation between cumulative dosage $D$
and $AUC$ can be derived for more complicated (multi-compartment) pharmacokinetics.

The solution of \eqref{pk} is given by
\begin{equation}\label{csol}C(t)=V^{-1}\int_0^t e^{-\mu(t-s)}d(s)ds.\end{equation}
We note that \eqref{csol} is meaningful even if
$d(t)$ is not a function, but is rather an
arbitrary non-negative measure - and may therefore include $\delta$-functions which represent 
{\it{bolus doses}}, that is a finite amount
of antimicrobial which is injected instantaneously, as we shall do below.
In any case, $C(t)$ given by \eqref{csol}
will be positive for all $t>0$ sufficiently large, so that the 
function $C_{opt}(t)$ given by \eqref{copt1} cannot be represented in the form \eqref{csol}, that is, it is not pharmacokinetically feasible.

We can, however, generate concentration
profiles which take a constant value
$\bar{c}$ for a duration $0\leq t\leq T_{bc}$ by
administering a bolus loading dose of size $V\bar{c}$,
at time $t=0$ to raise the concentration 
to $\bar{c}$, and thereafter supplying the drug as a continuous infusion at rate $\mu V\bar{c}$ up to time $T_{bc}$, so as to maintain the concentration $\bar{c}$. This is known as a {\it{bolus+continuous}} ($bc$) dosage schedule (\cite{derendorf}). Note that in order to achieve reduction in the microbial load
we must take $\bar{c}>zMIC$.
The expression for this dosing schedule is thus
\begin{equation}\label{ds}d_{bc}(t)=V\bar{c}\cdot [\delta(t)+\mu \cdot H(T_{bc}-t)],\end{equation}
where the $\delta$-function represents 
the bolus dose, and $H$ is the 
Heaviside function: $H(t)=0$ for $t<0$ and
$H(t)=1$ for $t\geq 0$.

The resulting concentration profile, given by the solution of \eqref{pk}, will be 
\begin{equation}\label{BC}C_{bc}(t)=\begin{cases}
		\bar{c} & t\leq T_{bc}\\
		\bar{c}e^{-\mu (t-T_{bc})}& t>T_{bc}
	\end{cases}.\end{equation}
We now formulate and study the problem of optimizing a bolus+continuous dosing schedule.

\begin{prob}\label{lc}
	Given a target value $LR_{target}$,
	find, among all dosing schedules of the form \eqref{ds}, parameterized by $\bar{c}$ and 
	$T_{bc}$, for which the corresponding 
	log-reduction is $LR_{max}[C_{bc}(t)]=LR_{target}$, the 
	one for which the $AUC$	is minimal.
\end{prob}

The solution of this problem is given by
\begin{theorem}\label{main2}
	Assume (A1)-(A3) and $\alpha>1$.
	Define
	\begin{equation}\label{dphi0}
		\rho=\int_{zMIC}^{c_{opt}}
		\frac{k(u)-r}{u}du,\end{equation}
	where $zMIC$ is the solution of \eqref{dzMIC} and $c_{opt}$ is the solution of \eqref{copte}.
	
	Then:
	
	(i) If $\rho<\ln(10)\cdot LR_{target}\cdot \mu$, then the solution of 
	Problem \ref{lc} is given by
	\begin{equation}\label{ds1}d_{opt}(t)=V\cdot c_{opt}\cdot [\delta(t)+\mu \cdot H(T_{bc,opt}-t)],\end{equation}
	where
	\begin{equation}\label{barts}T_{bc,opt}=\frac{\ln(10)\cdot LR_{target} -\mu^{-1}\rho}{k(c_{opt})-r}.\end{equation}
	The resulting concentration profile, given by the solution of \eqref{pk}, is
	\begin{equation}\label{BCopt}C_{bc,opt}(t)=\begin{cases}
			c_{opt} & t\leq T_{bc,opt}\\
			c_{opt}e^{-\mu (t-T_{bc,opt})}& t>T_{bc,opt}
		\end{cases}.\end{equation}
	The target log-reduction $LR_{target}$ of the microbial population will be achieved at time
	\begin{equation}\label{tso}T^*=T_{bc,opt}+\frac{1}{\mu}\cdot \ln\left(\frac{c_{opt}}{zMIC} \right),\end{equation}
	and the corresponding $AUC$ is
	\begin{equation}\label{auc2}AUC_{bc,opt}=\left[\mu^{-1}+ \frac{\ln(10)\cdot LR_{target}- \mu^{-1}\rho}{k(c_{opt})-r}\right]\cdot c_{opt}.\end{equation}
	
	(ii) If $\rho\geq \ln(10)\cdot LR_{target}\cdot \mu$, the solution of Problem \ref{lc} is given by
	$$d_{opt}(t)=Vc^*\delta(t),$$
	where $c^*$ is the solution of the equation
	\begin{equation}\label{cstar}\int_{zMIC}^{c^*}
		\frac{k(u)-r}{u}du=\ln(10)\cdot LR_{target}\cdot \mu,\end{equation}
	so that
	$$C_{bc,opt}(t)=c^*e^{-\mu t}.$$
	The target microbial population 
	will be reached at time
	$$T^*=\frac{1}{\mu}\cdot \ln\left(\frac{c^*}{zMIC} \right),$$
	and
	$$AUC_{bc,opt}=\mu^{-1}c^*.$$
\end{theorem}

We thus see that:

(i) If $\rho<\ln(10)\cdot LR_{target}\cdot \mu$, corresponding to sufficiently high decay rate of the antimicrobial, the optimal bolus+continuous dosing schedule maintains the {\it{same}} constant concentration $c_{opt}$ as the `ideal' concentration profile $C_{opt}(t)$ of Theorem \ref{mainr}, but for shorter
time duration. Indeed from \eqref{topte} and \eqref{barts} we have
\begin{equation}\label{dift}T_{opt}-T_{bc,opt}=\frac{\tau}{\ln(2)}\cdot\frac{\rho}{k(c_{opt})-r}>0.\end{equation}
We note also that, since $c_{opt}$ maximizes the function $f_r(c)$ given by \eqref{deff}, we have the inequality
\begin{equation}\label{inrho}\rho=\int_{zMIC}^{c_{opt}}\frac{k(u)-r}{u}du
	\leq (c_{opt}-zMIC)\cdot \frac{k(c_{opt})-r}{c_{opt}},\end{equation}
so that \eqref{dift} implies
$$0< T_{opt}-T_{bc,opt}\leq \frac{\tau}{\ln(2)}\cdot \left(1-\frac{zMIC}{c_{opt}}\right).$$
This implies that, as the antimicrobial half-life $\tau$ becomes short, $T_{bc,opt}$ 
converges to $T_{opt}$, so that the concentration profile induced by the optimal bolus+continuous schedule approaches the `ideal' optimal schedule $C_{opt}(t)$.
The $AUC$ achieved by the optimal bolus+continuous schedule will of course be higher than that obtained using the `ideal' concentration profile attaining the same log-reduction $LR_{target}$.
Indeed from \eqref{aucopt},\eqref{auc2} and \eqref{inrho} we have
\begin{equation}\label{difauc}AUC_{bc,opt}-AUC_{opt}=\frac{\tau}{\ln(2)}\cdot \left(1-\frac{\rho}{k(c_{opt})-r} \right)c_{opt}\geq \frac{\tau}{\ln(2)}\cdot zMIC.\end{equation}
As $\tau$ becomes small, \eqref{difauc} shows that the $AUC_{bc,opt}$ approaches $AUC_{opt}$.

(ii) If $\rho\geq \ln(10)\cdot LR_{target}\cdot \mu$, corresponding to a slow decay rate of the antimicrobial, the optimal dosing schedule consists of a single bolus dose raising the antimicrobial concentration to the value $c^*$
at time $t=0$. From \eqref{cstar} it follows that,
as $\mu\rightarrow 0$,
$$c^*=zMIC+\frac{\ln(10)\cdot LR_{target}}{k'(zMIC)}\cdot \mu +o(\mu),$$
so that for small $\mu$ (long antimicrobial half-life) the optimal bolus dose raises the concentration to slightly above the 
pharmacodynamic minimal inhibitory concentration $zMIC$.

To begin the analysis leading to Theorem \ref{main2}, we calculate the values 
$LR_{max}$ and $AUC$ corresponding to 
a bolus+continuous dosing schedule.

\begin{lemma}
	Consider a bolus+continuous schedule $d_{bc}(t)$ (see \eqref{ds}), with $\bar{c}>zMIC$, and the induced 
	antimicrobial concentration profile $C_{bc}(t)$ (see \eqref{BC}). Then: 
	
	(i) The maximal log-reduction corresponding to this 
	concentration profile is
	\begin{equation}\label{bclr}
		LR_{max}=\frac{1}{\ln(10)}\left[T_{bc}\cdot (k(\bar{c})-r)+\mu^{-1}\phi(\bar{c}) \right],
	\end{equation}
	where the function $\phi(c)$ is defined by:
	\begin{equation}\label{dphi}\phi(c)=\int_{zMIC}^{c}
		\frac{k(u)-r}{u}du.\end{equation}
\end{lemma}

(ii) The $AUC$ corresponding to this dosage schedule is
\begin{equation}\label{aucc}AUC=[\mu^{-1}+T_{bc}]\cdot \bar{c}.\end{equation}

\begin{proof}
	(i) By \eqref{LR}, the log-reduction 
	of the microbial load corresponding to 
	\eqref{BC}, at time $T$, is
	$$LR(T)=\frac{1}{\ln(10)}\cdot \int_0^T [k(C_{bc}(t))-r]dt,$$
	hence 
	$$T<T_{bc}\;\;\Rightarrow\;\; LR'(T)=k(C_{bc}(T))-r=k(\bar{c})-r>0,$$
	and
	$$\lim_{T\rightarrow \infty}LR'(T)=\lim_{T\rightarrow \infty}(k(C_{bc}(T))-r)=\lim_{T\rightarrow \infty}(k(\bar{c}e^{-\mu (T-T_{bc})})-r)=-r<0.$$
	We thus have that $LR(T)$ is an increasing function for $T<T_{bc}$, and a decreasing function for sufficiently large $T$, so that its maximum attained at some 
	$T^*>T_{bc}$ satisfying
	$LR'(T^*)=0$, that is 
	\begin{eqnarray*}k(C_{bc}(T^*))-r=0\;\;&\Leftrightarrow&\;\;
		C_{bc}(T^*)=zMIC
		\;\;\Leftrightarrow\;\;
		\bar{c}e^{-\mu (T^*-T_{bc})}=zMIC
		\\&\Leftrightarrow&\;\;
		T^*=T_{bc}+\frac{1}{\mu}\cdot \ln\left(\frac{\bar{c}}{zMIC} \right).
	\end{eqnarray*}
	Using the change of variable 
	$u=\bar{c}e^{-\mu(t-T_{bc})}$ in the integral below, we conclude that
	\begin{eqnarray*}\ln(10)\cdot LR_{max}&=&\ln(10)\cdot \max_{T>0}LR(T)=\ln(10)\cdot LR\left(T^* \right)\\
		&=&T_{bc}\cdot (k(\bar{c})-r)+\int_{T_{bc}}^{T^*} \left[k\left(\bar{c}e^{-\mu(t-T_{bc})}\right)-r\right]dt \\
		&=&T_{bc}\cdot (k(\bar{c})-r)+\frac{1}{\mu}\int_{\bar{c}e^{-\mu(T^*-T_{bc})}}^{\bar{c}} \frac{k(u)-r}{u}du \\
		&=&T_{bc}\cdot (k(\bar{c})-r)+\frac{1}{\mu}\int_{zMIC}^{\bar{c}} \frac{k(u)-r}{u}du  \\&=&T_{bc}\cdot (k(\bar{c})-r)+\mu^{-1}\phi(\bar{c}) ,
	\end{eqnarray*}
	where the function $\phi$ is defined by
	\eqref{dphi}.
	
	(ii) Using \eqref{aucd}, the $AUC$ corresponding to the concentration 
	profile generated by the dosing schedule 
	\eqref{ds} is given by
	$$AUC=V^{-1}\mu^{-1}\int_0^{\infty}d(t)dt=V^{-1}\mu^{-1}\left[V\bar{c}+\mu V\bar{c}T_{bc} \right]=[\mu^{-1}+T_{bc}]\cdot \bar{c}.$$
\end{proof}

\begin{proof}[Proof of Theorem \ref{main2}]
	
	By \eqref{bclr}, in order to achieve a given 
	log-reduction $LR_{target}$ using a dosing schedule of the form \eqref{ds}, the parameters 
	$\bar{c},T_{bc}$ defining this schedule 
	must satisfy the constraint $LR_{max}[C(t)]=LR_{target}$, or
	\begin{equation}\label{co1}T_{bc}\cdot (k(\bar{c})-r)+\mu^{-1}\phi(\bar{c}) =\ln(10)\cdot LR_{target}.\end{equation}
	We need to minimize the expression \eqref{aucc} over $(\bar{c},T_{bc})$, under the constraints 
	\eqref{co1} and \begin{equation}\label{ac}\bar{c}\geq zMIC,\;\;\; T_{bc}\geq 0.
	\end{equation}
	
	The constraint \eqref{co1} can be written
	as 
	\begin{equation}\label{bart}T_{bc} =\frac{\ln(10)\cdot LR_{target}- \mu^{-1}\phi(\bar{c})}{k(\bar{c})-r},\end{equation}
	and the inequality constraints \eqref{ac} imply that $\bar{c}$ must 
	satisfy 
	\begin{equation}\label{ine1}zMIC\leq \bar{c}\leq \phi^{-1}(\ln(10)\cdot LR_{target}\cdot \mu)=c^*,
	\end{equation}
	where $c^*$ is the solution of \eqref{cstar}.
	
	Substituting \eqref{bart} into \eqref{aucc} we get
	\begin{equation}\label{aucn}AUC=AUC(\bar{c})=\left[\mu^{-1}+ \frac{\ln(10)\cdot LR_{target}-\mu^{-1} \phi(\bar{c})}{k(\bar{c})-r}\right]\cdot \bar{c},\end{equation}
	which must be minimized over $\bar{c}$
	satisfying \eqref{ine1}. Noting that 
	the expression \eqref{aucn} goes to 
	$+\infty$ when $\bar{c}\rightarrow zMIC$,
	we see that the minimum is attained either at (a) an interior point of the interval \eqref{ine1}, or (b) at $\bar{c}=c^*$. 
	
	If $c_{opt}<c^*$, then 
	since $c_{opt}$ is the maximizer of $f_r(c)$
	given by \eqref{deff}, we have
	
	\begin{eqnarray*}&&\int_{c_{opt}}^{c^*}
		\frac{k(u)-r}{u}du<  (c^*-c_{opt})\cdot \frac{k(c_{opt})-r}{c_{opt}}\\
		&\Leftrightarrow&\;\; \phi(c^*)- \phi(c_{opt})< \frac{c^*-c_{opt}}{c_{opt}}\cdot (k(c_{opt})-r)\\
		&\Leftrightarrow&\;\;\left[\mu^{-1}+ \frac{\ln(10)\cdot LR_{target}- \mu^{-1}\phi(c_{opt})}{k(c_{opt})-r}\right]\cdot c_{opt}< \mu^{-1}c^* \\
		&\Leftrightarrow& AUC(c_{opt})<AUC(c^*),
	\end{eqnarray*}
	so that the minimum of $AUC(\bar{c})$ 
	in the interval \eqref{ine1} is attained at 
	an interior point, at which $AUC'(\bar{c})$
	must vanish, and using the fact that
	$\phi'(c)=\frac{k(c)-r}{c}$ we compute
	\begin{eqnarray*}AUC'(\bar{c})&=&
\;\;\frac{[\ln(10)\cdot LR_{target}- \mu^{-1}\phi(\bar{c})][k(\bar{c})-r-k'(\bar{c})\cdot \bar{c}]}{(k(\bar{c})-r)^2}=0\\&\Leftrightarrow&\;\;k(\bar{c})-r-k'(\bar{c})\cdot \bar{c}=0
		\;\;\Leftrightarrow\;\;\bar{c}=c_{opt}.\end{eqnarray*}
	Thus, from \eqref{bart} we get \eqref{barts}, and from \eqref{aucc} we get \eqref{auc2}. 
	
	On the other hand, if 
	$c_{opt}\geq c^*$, the above calculation shows that the derivative of $AUC(\bar{c})$ does not vanish in the interior of the interval \eqref{ine1}, so that the minimum is attained at $\bar{c}=c^*$, proving part (ii) of the theorem.
\end{proof}

\section{Optimal bolus+continuous dosing in the case of a Hill-type pharmacodynamic function}
\label{bchill}

We now apply the results of Theorem \ref{main2}
to the case in which $k(c)$ is a Hill-type function \eqref{hill}, and provide numerical
examples of the results obtained.

An explicit evaluation of the integral in \eqref{dphi} gives
$$\phi(c)=\frac{r}{\gamma}
\cdot 
\ln\left(\frac{\left(1+\left(\frac{c}{C_{50}}\right)^\gamma\right)^\alpha}{\left(\frac{c}{C_{50}}\right)^{\gamma }}\cdot \frac{(\alpha-1)^{\alpha-1}}{\alpha^\alpha}\right),$$
hence, using \eqref{copth},
$$\rho=\phi(c_{opt})=\frac{r}{\gamma}
\cdot 
\ln\left(\frac{\left(\frac{\gamma+1}{2} +\sqrt{\left(\frac{\gamma-1}{2}\right)^2+\frac{\gamma}{\alpha} }\right)^\alpha}{1+\alpha\cdot \left( \frac{\gamma-1}{2}  + \sqrt{\left(\frac{\gamma-1}{2} \right)^2+\frac{\gamma}{\alpha} }\right)}\right),$$
and $c^*$ is the solution of 
\begin{equation}\label{cse}
	\frac{\left(1+\left(\frac{c^*}{C_{50}}\right)^\gamma\right)^\alpha}{\left(\frac{c^*}{C_{50}}\right)^{\gamma }}=10^{LR_{target}\cdot \frac{\mu\gamma}{r}}\cdot \frac{\alpha^\alpha}{(\alpha-1)^{\alpha-1}}.\end{equation}
The condition
$\rho<\ln(10)\cdot LR_{target}\cdot \mu$ holds iff
\begin{equation}\label{ch}\tau<T_2\cdot \gamma\cdot \ln(10)\cdot LR_{target}
	\cdot \left[
	\ln\left(\frac{\left(\frac{\gamma+1}{2} +\sqrt{\left(\frac{\gamma-1}{2}\right)^2+\frac{\gamma}{\alpha} }\right)^\alpha}{  1+\alpha\cdot \left( \frac{\gamma-1}{2}  + \sqrt{\left(\frac{\gamma-1}{2} \right)^2+\frac{\gamma}{\alpha} }\right)}\right)\right]^{-1}.\end{equation}
We thus have:
\begin{itemize}
	\item If \eqref{ch} holds, that is the antimicrobial half-life is sufficiently short, then the solution of Problem \ref{lc}
	is the bolus+continuous schedule \eqref{ds1}, where
	$c_{opt}$ is given by \eqref{copth} and 
	$T_{bc,opt}$ is given by 
	\begin{eqnarray*}&&\frac{T_{bc,opt}}{T_2}=\frac{1}{\ln(2)(\alpha-1)}\cdot 
		\left(1+\frac{1}{\frac{\gamma-1}{2} + \sqrt{\left(\frac{\gamma-1}{2} \right)^2+\frac{\gamma}{\alpha} }}\right)\\ &\times&\left[\ln(10)\cdot LR_{target}-\frac{1}{\gamma}\cdot 
		\ln\left(\frac{\left(\frac{\gamma+1}{2} +\sqrt{\left(\frac{\gamma-1}{2}\right)^2+\frac{\gamma}{\alpha} }\right)^\alpha}{ 1+\alpha\cdot \left( \frac{\gamma-1}{2}  + \sqrt{\left(\frac{\gamma-1}{2} \right)^2+\frac{\gamma}{\alpha} }\right)}\right)\cdot \frac{\tau}{T_2}\right].\end{eqnarray*}
	Note that the value $T_{bc,opt}$ does not depend on the 
	half-saturation constant $C_{50}$.
	
	\item If the reverse inequality to \eqref{ch} holds, that is if the antimicrobial half-life is sufficiently long, then  the solution of Problem \ref{lc} is a single bolus dose $Vc^*$, where 
	$c^*$ is given by \eqref{cse}.
	
\item In the special case $\gamma=1$ (the $E_{max}$ model), the expression for 
$T_{bc,opt}$ reduces to
\begin{eqnarray*}\gamma=1\;\;\Rightarrow\;\;\frac{T_{bc,opt}}{T_2}=\frac{1}{\ln(2)(\sqrt{\alpha}-1)}\times\left[\ln(10)\cdot LR_{target}-
	\ln\left(\frac{\left(\sqrt{\alpha} +1\right)^{\alpha-1}}{\sqrt{\alpha^\alpha }}\right)\cdot \frac{\tau}{T_2}\right],\end{eqnarray*}
and the condition \eqref{ch}
reduces to positivity of the value on the right-hand side.

\end{itemize}

Tables 3,4 present numerical results regarding optimal bolus+continuous schedules, using the above formulae for parameter values which are in a range relevant to applications, which can be compared with the 
results concerning the `ideal' concentration profile
in Table 2. In Table 3 
it is assumed that the ratio of the antimicrobial 
half-life to the microbial doubling time is $\frac{\tau}{T_2}=4$, while in Table 4
we take faster antimicrobial decay, $\frac{\tau}{T_2}=2$. For all parameter values considered, the condition \eqref{ch} holds, so that the optimal schedule includes both a bolus and a continuous infusion. Comparing the $AUC$ obtained in Table 3 with the `ideal' ones in Table
2, we observe that, although, as expected, the values $AUC_{bc,opt}$attained by the optimal bolus+continuous schedules are higher than $AUC_{opt}$, in most cases they lie within $25\%$ of that value, with the exception of extreme cases of high Hill coefficient $\gamma$ and antimicrobial potency $\alpha$ (e.g. for $\gamma=5,\alpha=6$, $AUC_{bc,opt}$ is $67.5\%$ higher than $AUC_{opt}$).
For shorter antimicrobial half-lives, as in Table 4,  $AUC_{bc,opt}$ is even closer to $AUC_{opt}$. In general, we can conclude that 
for realistic parameter values, the optimal bolus+continuous dosing schedule attains outcomes which are quite close to the `ideal' one.
As an example, in Figure \ref{exa} we compare the 
`ideal' concentration curve and the concentration curve induced by the optimal bolus+continuous schedule, and the corresponding microbial population curves, using parameters fitting the antibacterial Tobramycin applied to {\it{Pseudomonas aeruginosa ATCC 27853}}, as in the example of \cite{bouvier} (see figure caption for parameter values). The 
duration of infusion in the optimal bolus+continuous schedule is $T_{bc,opt}=31.3$ hours, slightly shorter than the duration  $T_{opt}=33.0$ of the `ideal' concentration profile. For the 
`ideal' concentration profile, the bacterial population reaches the target (eradication) value $10^{-7}B_0$ at $t=T_{opt}$, while 
for the optimal bolus+continuous schedule 
the target value is reached at time $T^*=34.8$ hours (as given by \eqref{tso}), that is $3.5$ hours after 
antimicrobial infusion is ended.
The $AUC$ corresponding to the optimal bolus+continuous schedule is $5.82 \frac{mg\cdot h}{L}$, only $4\%$ higher than the `ideal' $AUC_{opt}=5.59 \frac{mg\cdot h}{L}$.

\begin{table}
	\centering
	
	\begin{minipage}{.4\linewidth}
		\centering
		{\bfseries\strut Optimal duration $\frac{T_{opt}}{T_2}$}
		\begin{tabular}{|l||*{5}{c|}}\hline
			\backslashbox{$\gamma$}{$\alpha$}
			&\makebox[2em]{2.0}&\makebox[2em]{3.0}&\makebox[2em]{4.0}
			&\makebox[2em]{5.0}&\makebox[2em]{6.0}\\\hline\hline
		0.5 & 95.7 & 59.5 & 47.0 & 40.7 & 36.7\\\hline
		1.0 & 53.5 & 28.9 & 20.2 & 15.7 & 12.8\\\hline
		2.0 & 35.4 & 17.1 & 10.8 & 7.5 & 5.5\\\hline
		3.0 & 30.3 & 14.3 & 8.8 & 6.0 & 4.3\\\hline
		4.0 & 28.0 & 13.2 & 8.1 & 5.5 & 4.0\\\hline
		5.0 & 26.8 & 12.6 & 7.8 & 5.3 & 3.9\\\hline	 
		\end{tabular}

	\end{minipage}%
	\quad % ----------------------------------
	\bigskip
	
	\begin{minipage}{.4\linewidth}
		\centering
		{\bfseries\strut Optimal AUC $\frac{AUC_{opt}}{T_2\cdot C_{50}}$}
		\begin{tabular}{|l||*{5}{c|}}\hline
			\backslashbox{$\gamma$}{$\alpha$}
			&\makebox[2em]{2.0}&\makebox[2em]{3.0}&\makebox[2em]{4.0}
			&\makebox[2em]{5.0}&\makebox[2em]{6.0}\\\hline\hline
			0.5 & 265.7 & 46.4 & 17.6 & 9.0 & 5.5\\\hline	
			1.0 & 143.1 & 47.4 & 26.0 & 17.3 & 12.8\\\hline	
			2.0 & 84.7 & 37.6 & 24.3 & 18.2 & 14.7\\\hline	
			3.0 & 66.1 & 32.2 & 22.0 & 17.2 & 14.3\\\hline	
			4.0 & 57.1 & 29.1 & 20.5 & 16.3 & 13.8\\\hline	
			5.0 & 51.7 & 27.1 & 19.4 & 15.6 & 13.4\\\hline	
				\end{tabular}

	\end{minipage}%
	\quad % ----------------------------------
	\bigskip 
	
			\caption{Optimal bolus+continuous dosage plans for achieving log-reduction $LR_{target}=7$, when the half-life of the antimicrobial satisfies $\frac{\tau}{T_2}=4$, for different values of $\alpha=\frac{k_{max}}{r}$ and of the Hill exponent $\gamma$.
			The concentration $c_{opt}$ to be maintained is the same as is Table 2.
			Top: Optimal duration of dosing, 
			Bottom: value of $AUC$ attained by the optimal schedule}
	
\end{table}

\begin{table}
	\centering
	
	\begin{minipage}{.4\linewidth}
		\centering
		{\bfseries\strut Optimal duration $\frac{T_{opt}}{T_2}$}
		\begin{tabular}{|l||*{5}{c|}}\hline
			\backslashbox{$\gamma$}{$\alpha$}
			&\makebox[2em]{2.0}&\makebox[2em]{3.0}&\makebox[2em]{4.0}
			&\makebox[2em]{5.0}&\makebox[2em]{6.0}\\\hline\hline			
			0.5 & 97.1 & 61.0 & 48.6 & 42.2 & 38.3\\\hline
			1.0 & 54.8 & 30.3 & 24.6 & 17.2 & 14.4\\\hline
			2.0 & 36.5 & 18.4 & 12.1 & 8.9 & 6.9\\\hline
			3.0 & 31.3 & 15.4 & 10.0 & 7.2 & 5.5\\\hline
			4.0 & 28.9 & 14.1 & 9.1 & 6.6 & 5.0\\\hline
			5.0 & 27.6 & 24.1 & 8.7 & 6.3 & 4.8\\\hline				 
		\end{tabular}

	\end{minipage}%
	\quad % ----------------------------------
	\bigskip
	
	\begin{minipage}{.4\linewidth}
		\centering
		{\bfseries\strut Optimal AUC $\frac{AUC_{opt}}{T_2\cdot C_{50}}$}
		\begin{tabular}{|l||*{5}{c|}}\hline
			\backslashbox{$\gamma$}{$\alpha$}
			&\makebox[2em]{2.0}&\makebox[2em]{3.0}&\makebox[2em]{4.0}
			&\makebox[2em]{5.0}&\makebox[2em]{6.0}\\\hline\hline
			0.5 & 261.8 & 45.4 & 17.2 & 8.8 & 5.3\\\hline
			1.0 & 139.3 & 45.4 & 24.6 & 16.3 & 11.9\\\hline
			2.0 & 81.1 & 34.8 & 22.0 & 16.2 & 12.8\\\hline
			3.0 & 62.6 & 29.3 & 19.4 & 14.7 & 12.0\\\hline
			4.0 & 53.7 & 26.1 & 17.7 & 13.6 & 11.2\\\hline
			5.0 & 48.4 & 24.1 & 16.6 & 12.9 & 10.7\\\hline
			
			\end{tabular}

	\end{minipage}%
	\quad % ----------------------------------
	\bigskip 
	
			\caption{Optimal bolus+continuous dosage plans for achieving log-reduction $LR_{target}=7$, when the half-life of the antimicrobial satisfies $\frac{\tau}{T_2}=2$, for different values of $\alpha=\frac{k_{max}}{r}$ and of the Hill exponent $\gamma$.
			The concentration $c_{opt}$ to be maintained is the same as is Table 2.
			Top: Optimal duration of dosing, 
			Bottom: value of $AUC$ attained by the optimal schedule}
	
\end{table}

\begin{figure}
	\begin{center}
		\includegraphics[width=0.45\linewidth]{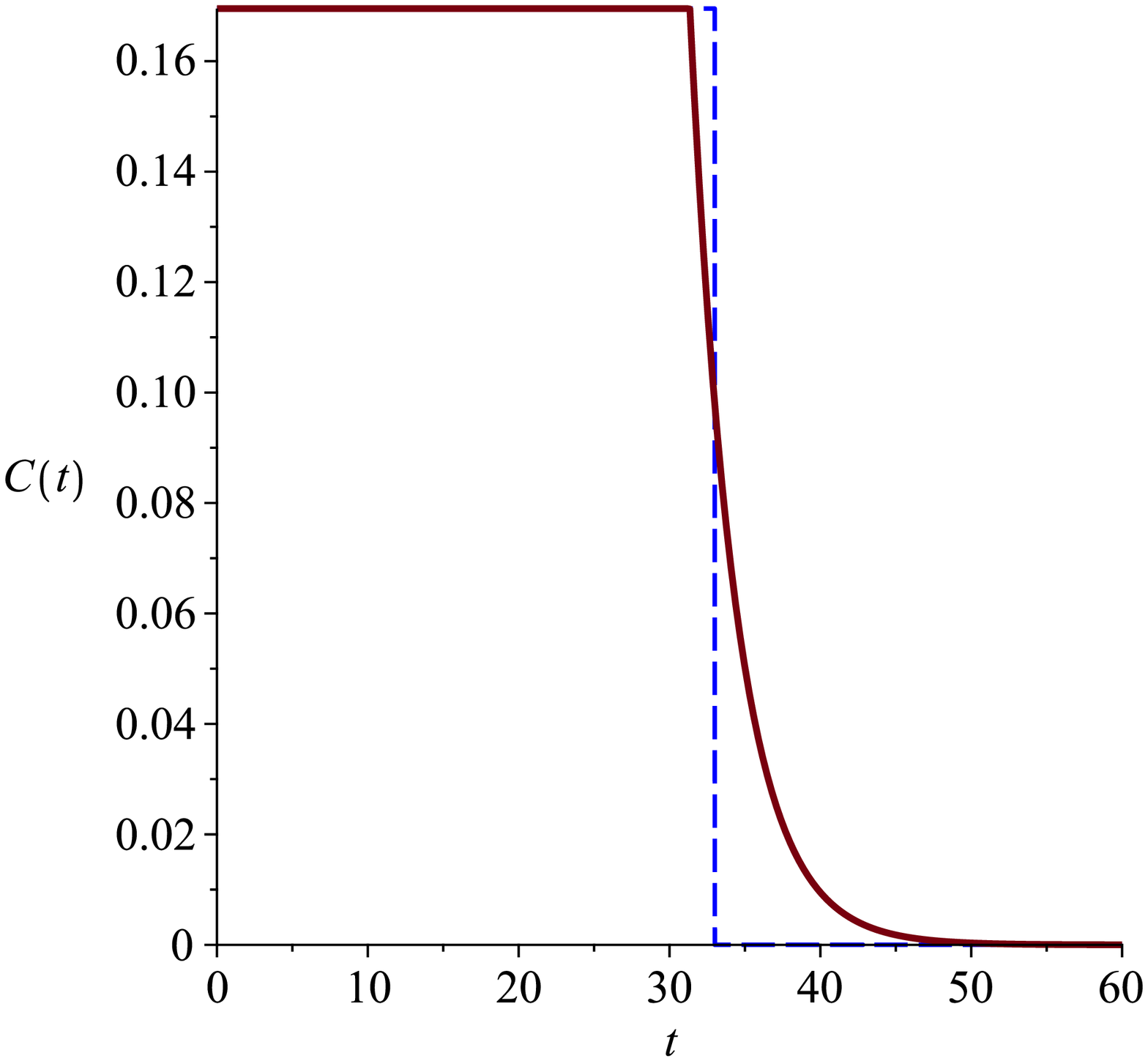}
		\includegraphics[width=0.45\linewidth]{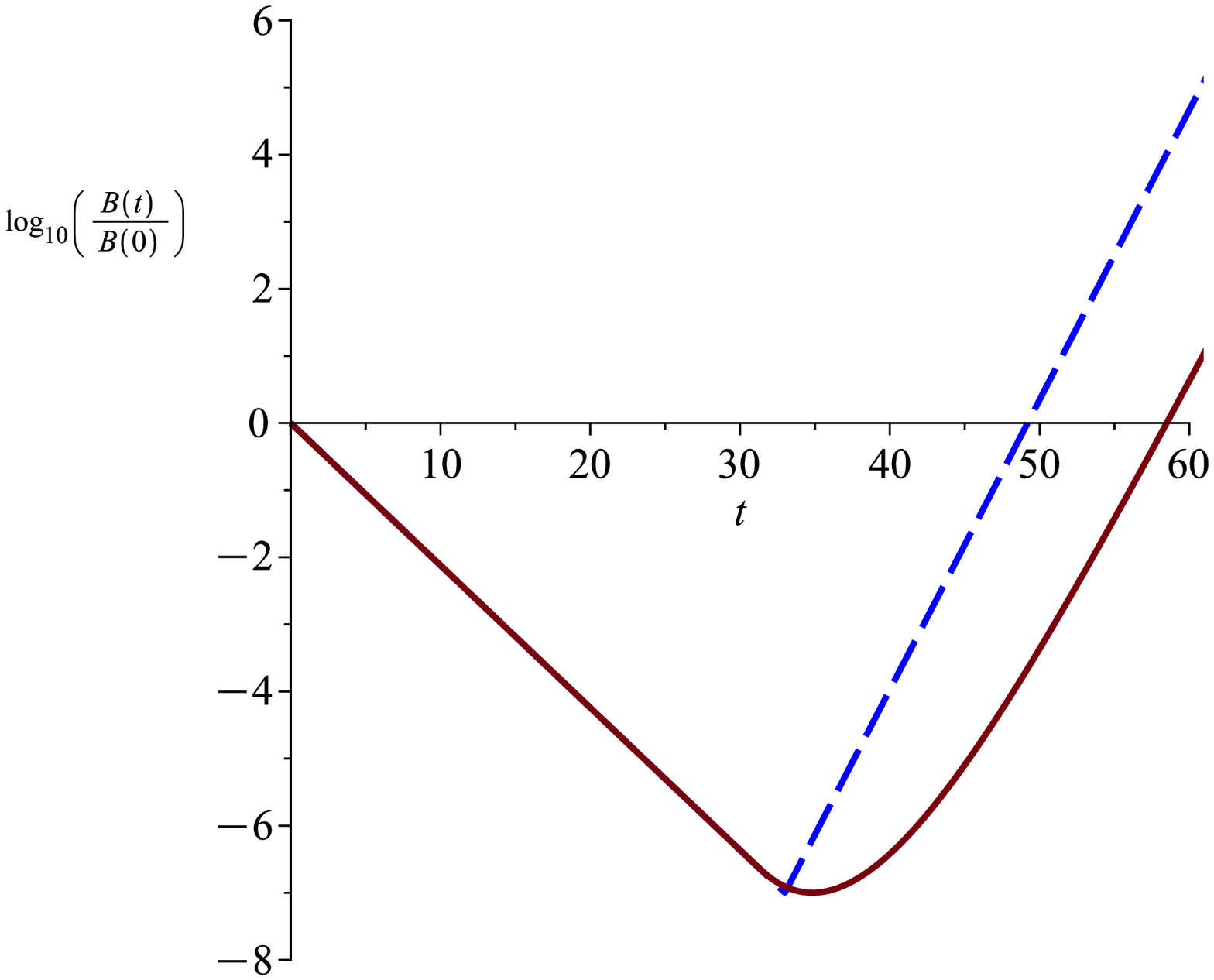}
	\end{center}
	\caption{Left: the `ideal' optimal antimicrobial concentration profile (blue dashed line), and the concentration $C(t)$ (in $\frac{mg}{L}$) at time $t$ (in hours) induced by the optimal bolus+continuous schedule (red line), Right: The corresponding microbial populations relative to the initial population. Parameters, corresponding to the example from \cite{bouvier} of Tobramycin applied to {\it{Pseudomonas aeruginosa}}, are: $C_{50}=4.187 \frac{mg}{L}$, $k_{max}=7.115 h^{-1}$, $\gamma=0.416$, $r=0.995 h^{-1}$, $\mu=0.333 h^{-1}$, $LR_{target}=7$. The `ideal' concentration profile consists of the concentration $c_{opt}=0.17 \frac{mg}{L}$ for a time duration $T_{opt}=33.0h$, while the optimal bolus+continuous dosing schedule provides a bolus dose of size $c_{opt}\cdot V=0.17\cdot V mg$ (where $V$ is the volume of distribution), followed by a continuous infusion of $c_{opt}\cdot \mu\cdot V=0.056\cdot V \frac{mg}{h}$ up to time $T_{bc,opt}=31.3h$. }
	\label{exa}
\end{figure}

\begin{figure}
	\begin{center}
		\includegraphics[width=0.45\linewidth]{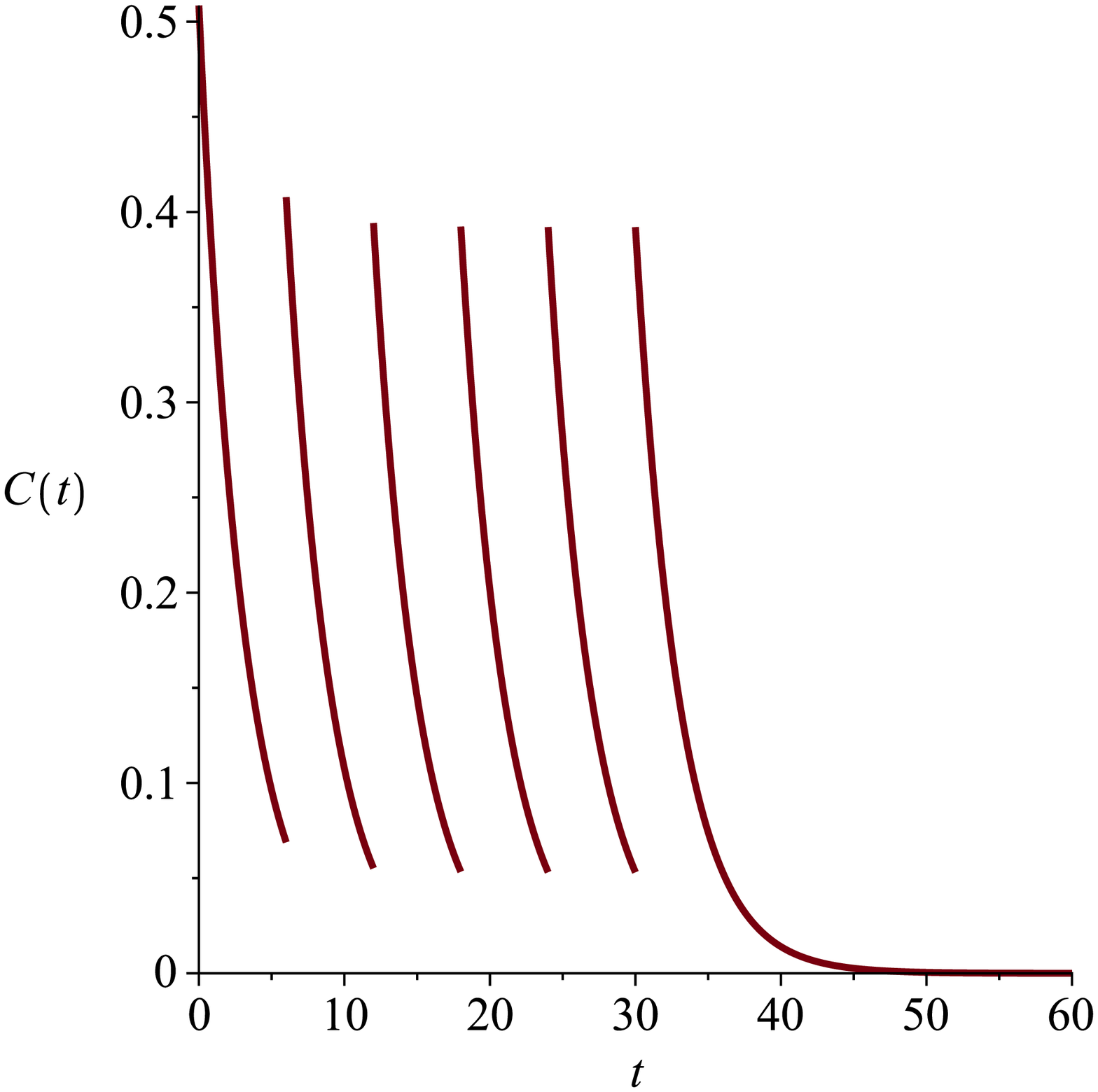}
		\includegraphics[width=0.45\linewidth]{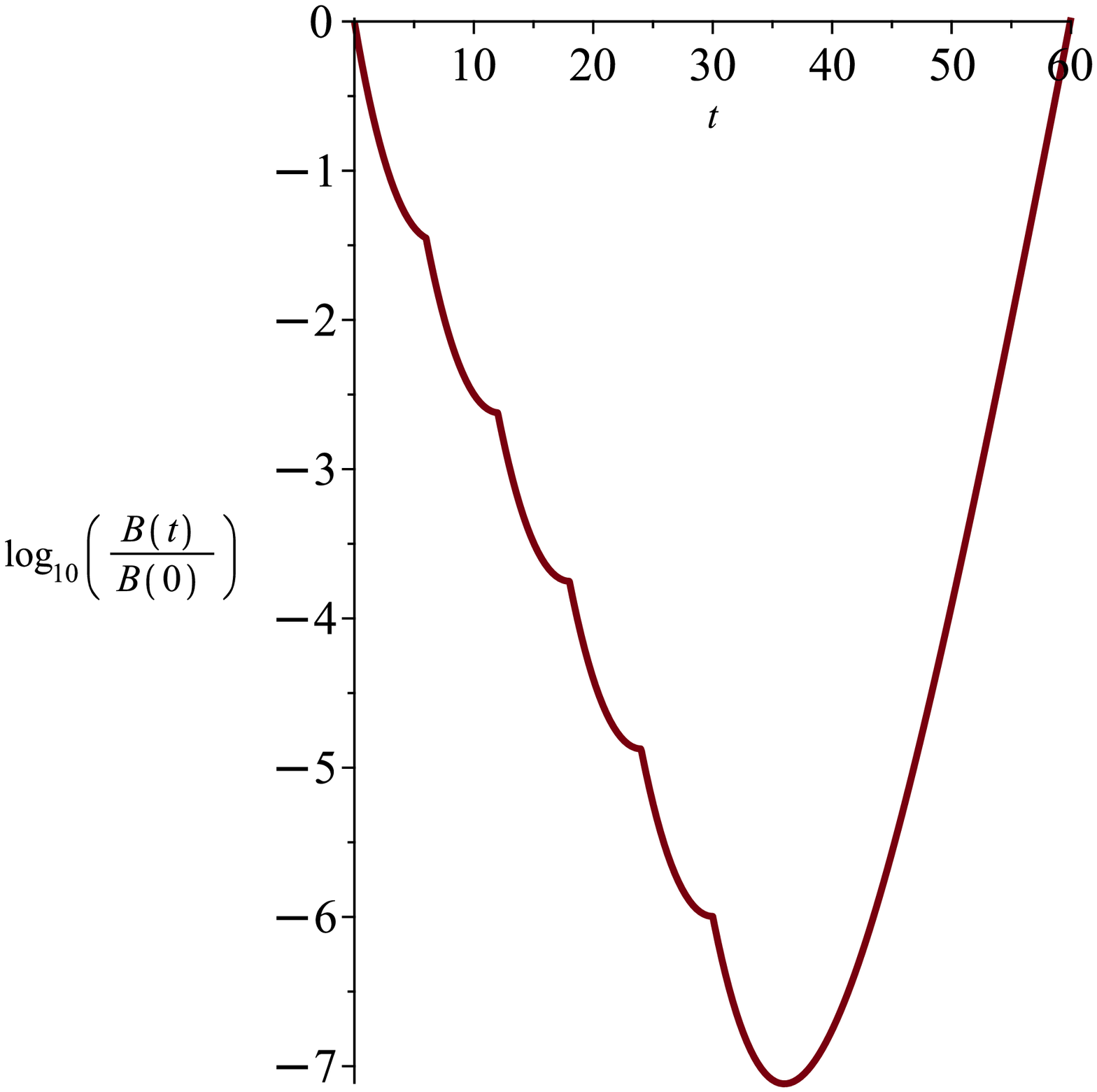}\\
		\includegraphics[width=0.45\linewidth]{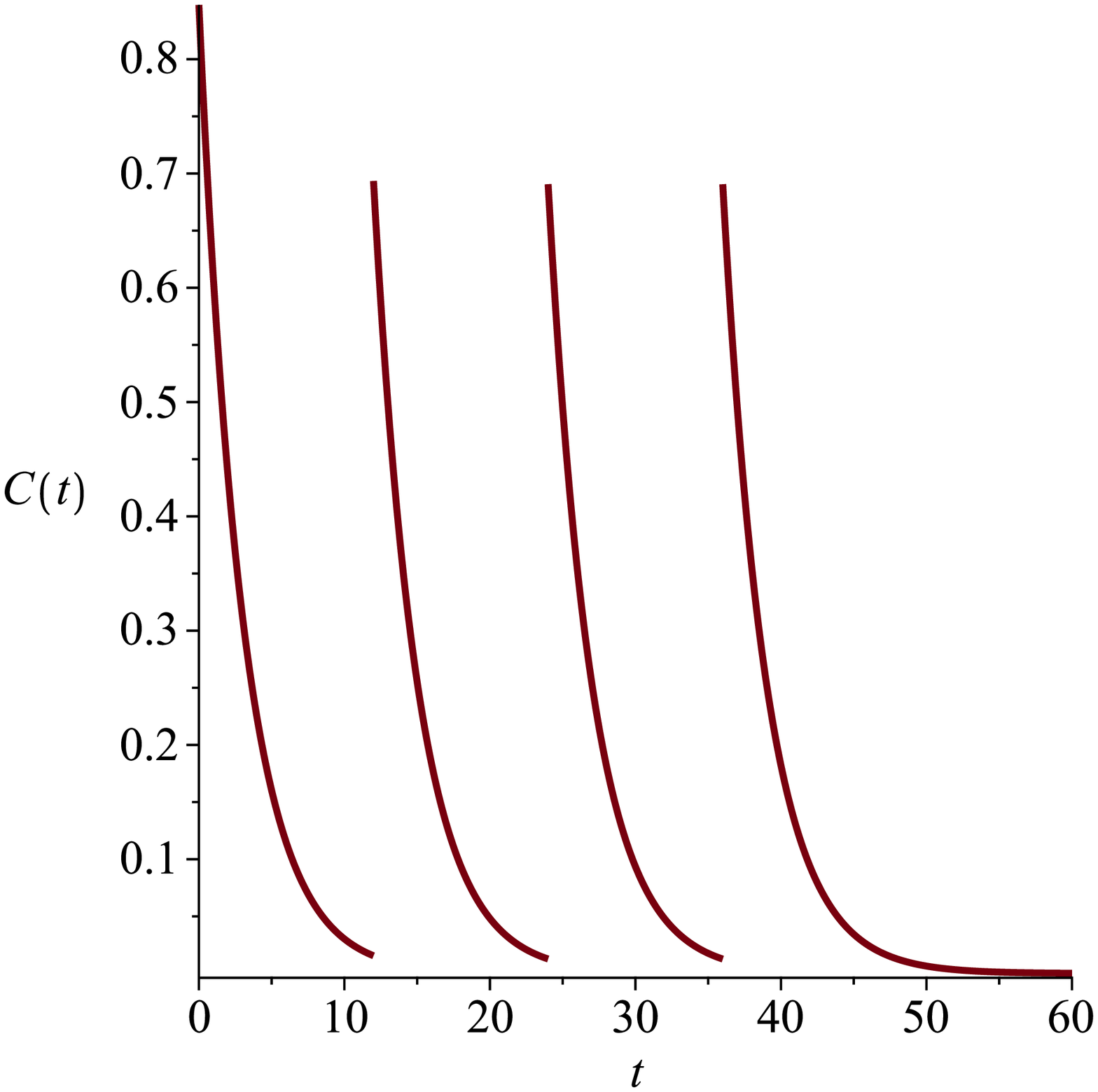}
		\includegraphics[width=0.45\linewidth]{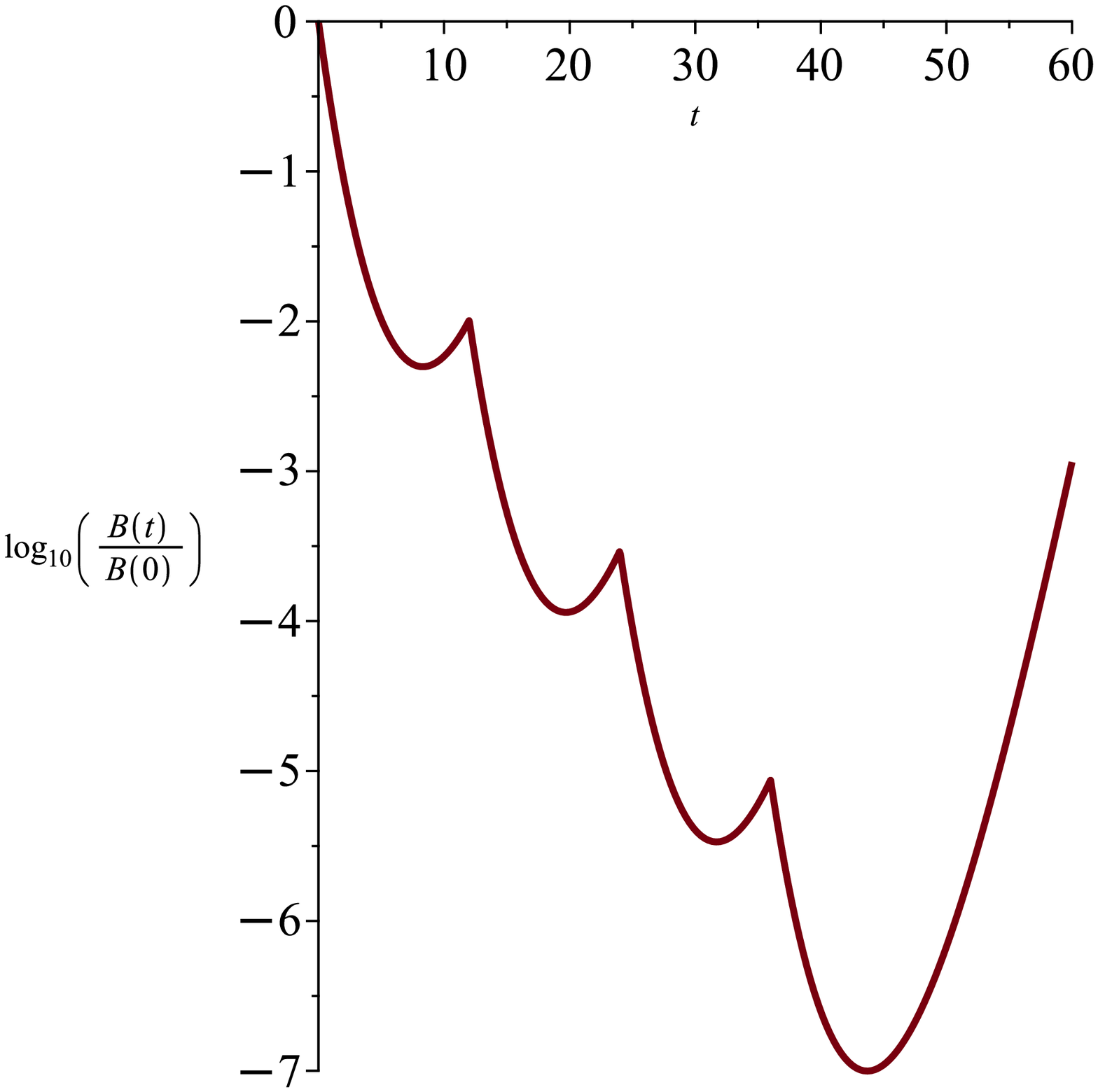}
	\end{center}
	\caption{Outcomes of a intermittent dosing schedules approximating the optimal bolus+continuous schedule presented in Figure \ref{exa} (antimicrobial concentration $C(t)$ ($\frac{mg}{L}$) at time $t$ (hours) on the left, microbial population on the right).
	Parameters are identical to those in Figure \ref{exa} (corresponding to the example from \cite{bouvier} of Tobramycin applied to {\it{Pseudomonas aeruginosa}}). 
	The dosing schedule are: Top: a loading dose of size $0.51\cdot V mg$ (where $V$ is the volume of distribution) and $5$ bolus maintenance doses of size $0.34\cdot V mg$, given at $6h$ intervals. 
	Bottom: a loading dose of size $0.85\cdot V mg$ (where $V$ is the volume of distribution) and $3$ bolus maintenance doses of size $0.68\cdot V mg$, given at $12h$ intervals.}
	\label{exa2}
\end{figure}

In practice, the administration of a continuous dosing schedule requires the use of intravenous infusion, infusion pumps, or sustained/controlled release formulations. An intermittent dosing schedule, involving a series of bolus doses, is often opted for (\cite{derendorf}). It is 
intuitively obvious, however, and can be formally proved, that sufficiently frequent intermittent infusions, with appropriate doses, can be used to approximate a bolus+continuous schedule to with arbitrary precision, so that our results concerning the optimal bolus+continuous schedule are also relevant to the design of intermittent schedules, and in particular can be used to assess the degree to which a proposed intermittent schedule can be improved upon by increasing the frequency of dosing or by shifting to continuous infusion.
To illustrate this, we have taken the optimal 
	bolus+continuous schedule presented in Figure \ref{exa}, and generated an intermittent dosing schedule which approximates it, with bolus doses given at $6$-hour intervals. 
	The maintenance doses are equal to the total doses given by the optimal 
	bolus+continuous schedule over $6$ hours, while to obtain the loading dose we add the loading dose for the bolus+continuous schedule to the above maintenance dose. The resulting dosing schedule, whose outcomes are presented in Figure \ref{exa2} (top) achieves the target $B(t)=10^{-7}B_0$ after $36.0$ hours, compared with $34.8$ hours for the optimal bolus+continuous schedule with $AUC=6.61\frac{mg\cdot h}{L}$, only $14\%$ higher than the $AUC$ corresponding to optimal bolus+continuous schedule and $\%18$ higher than the `ideal' value $AUC_{opt}$. This demonstrates that it is possible to use the results obtained here to generate intermittent schedules which achieve outcomes quite close to the theoretical ideal. However, if the interval between doses becomes larger, the performance of the intermittent approximation to the optimal bolus+continuous schedule degrades. In the bottom part of Figure \ref{exa2} we used the same procedure to generate a schedule with $12$-hours intervals. This schedule achieves the target at time $t=43.4$ hours, and its $AUC$ is $8.64 \frac{mg\cdot h}{L}$, 
	$48\%$ higher than the `ideal' $AUC_{opt}$.

\section{Discussion}
\label{discussion}

The results obtained in this work provide a 
a baseline and a reference point for evaluating the efficient use of antimicrobials. Theorem \ref{mainr} 
describes the `ideal' concentration profile
leading to eradication of the microbial population, with a minimal $AUC$ - which consists of a constant concentration 
$c_{opt}$ applied over a period of duration $T_{opt}$.
We provided simple equations allowing to compute the key quantities 
$c_{opt}$ and $T_{opt}$ for an arbitrary
pharmacodynamic function $k(c)$, and explicit expressions for these quantities in the case of the widely-used Hill-type function (see \eqref{copth},\eqref{topth}).

Since the `ideal' concentration profile
is not strictly feasible due to 
pharmacokinetic contraints, we have also considered the question of determining
an optimal bolus+continuous dosing schedule, assuming first order pharmacokinetics. 
Our results show that the optimal 
dosing leads to the {\it{same}} constant concentration $c_{opt}$ as for the `ideal' concentration profile during a dosing period 
$T_{bc,opt}<T_{opt}$. Our numerical comparisons show that the results obtained 
using this optimal bolus+continuous dosing plan are in most cases only slightly 
inferior to those obtained using the `ideal' 
concentration profile. 

We note that while the `ideal' concentration profile was proved to be optimal among {\it{all}} concentration profiles (Theorem \ref{mainr}), in the investigation of dosing plans under 
pharmacokinetic constraints we restricted ourselves in advance to bolus+continuous plans with 
constant dosing rate following the loading dose. In fact we conjecture that no 
dosing schedule with a non-constant dosing rate can improve upon the performance of the dosing plans considered, but leave a full treatment of this question to future work. 

The models which we employed in this study are 
standard ones, which are widely applied in 
the quantitative literature on antimicrobial pharmacology. 
However, as always with mathematical modelling, it
is important to take into account the limitations of the model employed, and their possible implications regarding the conclusions drawn using the model. 
In particular, since a central feature of the `ideal' concentration profile, as presented in Theorem \ref{mainr}, is that it consists of a {\it{constant}} concentration provided over a finite duration of time, one should investigate whether additional mechanisms, not taken into account by the model, might modify this conclusion. Are there mechanisms whose inclusion in the model would entail that the `ideal' concentration profile will be time-varying, and if so in what way? 
Could an intermittent dosing schedule be optimal under some circumstances? Below we 
mention several relevant mechanisms, which have been included in some models of antimicrobial 
pharmacodynamics, and whose implications for 
the optimization of dosing schedules merit  systematic study.

(1) Density-dependence of microbial population growth. This will be relevant if, before treatment begins, the microbial  population reaches a size which is close to the environmental carrying capacity, so that resource limitation leads to significant reduction of the growth rate. Several published models have included this mechanism, either 
by replacing the constant per capita growth
term in \eqref{model} by a nonlinear (typically logistic) term (\cite{bhagunde,geli,kesisoglou,nikolaou1,paterson,tindall}), or by explicitly modelling the resource dynamics (\cite{ali,khan,levin,zilonova}). We expect that when density-dependent growth is included in the model, the `ideal' concentration profile (the solution of Problem \ref{prmain}) will no longer be constant in time. However, some simulations we have carried out (not shown) indicate that practical implications of this fact are not necessarily large: as soon as the microbial population is reduced, by the antimicrobial action, so that it is not close to the carrying capacity, the linear model as in \eqref{model}  becomes a good approximation, hence the associated `ideal' concentration profile \eqref{copt1} is nearly optimal. 
This issue should be systematically explored in future work.

(2) Immune response. 
When the strength of immune supression of microbial growth is comparable to that of the antimicrobial effect, it should be taken into consideration when modelling the microbial dynamics.
In some cases immune response is modelled 
as a killing term of microbes with a 
constant per capita rate, which is
equivalent to reducing the growth-rate parameter $r$ in \eqref{model}, hence does not induce a change in the model structure (\cite{goranova}).
However, in some studies the immune reponse is modelled as depending nonlinearly on the bacterial population size (\cite{hoyle}) or as an explicit time-varying component representing the gradual build-up of immune response
(\cite{geli,tindall}). Implications of such a change to model structure with respect to the optimal design of treatment require study.

(3) Persister cells. 
One mechanism which can allow some bacterial species to survive antibiotics is phenotypic switching of cells from an antibiotic-sensitive state to a 
persister state with greatly reduced sensitivity to the antibiotic, and for which growth rate is also reduced (\cite{brauner}). The small population of persistent cells is 
therefore at a competitive disadvantage under antibiotic-free conditions, but in the presence of antibiotic it allows the bacterial population to survive.
Theorerical studies of antimicrobial 
treatment of microbial populations including persisters (\cite{cogan}, \cite{cogan1}) show that, when the transition of microbes from the persistent state back to the sensitive state does not occur in the presence of the antimicrobial, the microbial population 
cannot be eliminated by a constant application of antimicrobial, and elimination requires a series of applications of the antimicrobial, separated by withdrawl periods. The optimal antibacterial treatment schedules  
in the presence of persisters will thus of necessity be quite different from those 
presented in this work.

(4) Antimicrobial resistance. This refers to clones of microbial cells which are insestive to antimicrobial action. The effort to prevent antimicrobial resistance is an important motivation for the efficient use of antimicrobial agents, which has been studied here. However the model we employed does not directly address this issue, as do models which include a 
separate compartment for resistant strains which may arise due to mutations (\cite{ali,geli,goranova,khan,lipsitch,marrec,morsky,nielsen,paterson,singh,zilonova}). This can lead to additional considerations regarding optimization
of antimicrobial use.

We thus conclude that it is important to examine whether, how, and to what extent
the incorporation of additional mechanisms 
into the models modifies the conclusions 
obtained here. The analytical results presented here should form a useful point of reference for comparison with results obtained using modified and extended models.

It is relevant to mention here the extensive literature on the modelling and optimization of chemotherapeutic treatment in oncology (see \cite{shi} for an overview), noting that some of the basic models employed in that field have a structure similar to those used in the study of antimicrobial treatment, with microbial population size replaced by tumor size. In particular, the recent works \cite{Leszczynski}, \cite{Ledzewicz} 
treat optimization problems which are closely related to those considered here, for a more general model including nonlinear growth, both in a purely pharmacodynamic version and 
in the case that pharamcokinetics is taken into account. These works derived first-order optimalilty conditions using the Pontryagin maximum principle, and used these conditions to obtain information on the structure of optimal solutions, though they did not present explicit and detailed results for the case of linear growth, as obtained here. We note that in the present work we did not use variational tools such as the maximum principle, and our proof of the (global) optimality of the `ideal' concentration profile (Theorem \ref{mainr}) is direct and elementary.

Finally, we remark that our analyses assume 
known pharmacodynamic and pharmacokinetic parameters, but these parameters themselves will vary among individuals, and the design of treatment plans must also take 
into account this heterogeneity, a fact which has led to the development of population PK/PD modelling (\cite{develde,vinks}). Incorporating the insights obtained from the results in the present work into the population perspective remains a task for future research.

{}


\vspace{1cm}




\begin{thebibliography}{1}
	
	\bibitem[Ali et al. (2022)]{ali}
	Ali, A, Imran, M, Sial, S, Khan, A (2022) Effective antimicrobial dosing in the presence of resistant strains. PLoS ONE 17:e0275762. https://doi.org/10.1371/journal.pone.0275762
	
	\bibitem[Austin et al. (1998)]{austin}
	Austin, DJ, White, NJ, Anderson, RM (1998) The dynamics of drug action on the within-host population growth of infectious agents: melding pharmacokinetics with pathogen population dynamics. J. Theor. Biol. 194:313-339. https://doi.org/10.1006/jtbi.1997.0438
	
	\bibitem[Bhagunde et al. (2015)]{bhagunde}
	Bhagunde, PR, Nikolaou, M, Tam, VH (2015) Modeling heterogeneous bacterial populations exposed to antibiotics: The logistic dynamics case. AIChE J. 61:2385-2393.  https://doi.org/10.1002/aic.14882
	
	\bibitem[Bouvier d'Yvoire and Maire (1996)]{bouvier}
	Bouvier d'Yvoire, MJ, Maire, PH (1996) Dosage regimens of antibacterials. Clin. Drug Invest. 11:229-239. https://doi.org/10.2165/00044011-199611040-00006
	
	\bibitem[Brauner et al. (2016)]{brauner}
	Brauner, A, Fridman, O, Gefen, O, Balaban, NQ (2016) Distinguishing between resistance, tolerance and persistence to antibiotic treatment. Nat. Rev. Microbiol., 14:320-330.
	
	\bibitem[Bulitta et al. (2019)]{bulitta}
	Bulitta, JB, Hope, WW, Eakin, AE, Guina, T., Tam, VH, Louie, A, Drusano, GL, Hoover, JL (2019) Generating robust and informative nonclinical in vitro and in vivo bacterial infection model efficacy data to support translation to humans. Antimicrob. Agents Chemother. 63:e02307-18.  https://doi.org/10.1128/AAC.02307-18
	
	
	\bibitem[Cicchese et al. (2017)]{cicchese}
	Cicchese, JM, Pienaar, E, Kirschner, DE, Linderman, JJ (2017) Applying optimization algorithms to tuberculosis antibiotic treatment regimens. Cell. Mol. Bioeng. 10:523-535. https://doi.org/10.1007/s12195-017-0507-6
	
	\bibitem[Cogan (2006)]{cogan}
	Cogan, NG (2006) Effects of persister formation on bacterial response to dosing. J. Theor. Biol. 238:694-703.
	https://doi.org/10.1016/j.jtbi.2005.06.017
	
	\bibitem[Cogan et al. (2012)]{cogan1}
	Cogan, NG, Brown, J, Darres, K,  Petty, K (2012) Optimal control strategies for disinfection of bacterial populations with persister and susceptible dynamics. Antimicrob. Agents Chemother. 56:4816-4826.
	https://doi.org/10.1128/aac.00675-12
	
	\bibitem[Colin et al. (2020)]{colin}
	Colin, PJ, Eleveld, DJ, \& Thomson, AH (2020) Genetic Algorithms as a Tool for Dosing Guideline Optimization: Application to Intermittent Infusion Dosing for Vancomycin in Adults. CPT: Pharmacometrics Syst. Pharmacol. 9:294-302. https://doi.org/10.1002/psp4.12512
	
	\bibitem[Corvaisier et al. (1998)]{corvaisier}
	Corvaisier, S, Maire, PH, Bouvier d'Yvoire, MY, Barbaut, X, Bleyzac, N, Jelliffe, RW (1998) Comparisons between antimicrobial pharmacodynamic indices and bacterial killing as described by using the Zhi model. Antimicrob. Agents Chemother. 42:1731-1737.  https://doi.org/10.1128/AAC.42.7.1731
	
	
	\bibitem[Czock and Keller (2007)]{czock}
	Czock, D, \& Keller, F (2007) Mechanism-based pharmacokinetic-pharmacodynamic modeling of antimicrobial drug effects. J. Pharmacokinet. Pharmacodyn. 34:727-751. https://doi.org/10.1007/s10928-007-9069-x
	
	\bibitem[Derendorf and Schmidt (2019)]{derendorf}
	Derendorf, H, Schmidt, S (2019) Rowland and Tozer's clinical pharmacokinetics and pharmacodynamics: concepts and applications, Wolters Kluwer, South Holland.
	
	\bibitem[de Velde at al. (2018)]{develde}
	de Velde, F, Mouton, JW, de Winter, BC, van Gelder, T, Koch, BC (2018) Clinical applications of population pharmacokinetic models of antibiotics: Challenges and perspectives. Pharmacol. Res. 134:280-288. https://doi.org/10.1016/j.phrs.2018.07.005
	
	\bibitem[Geli et al. (2012)]{geli}
	Geli, P, Laxminarayan, R, Dunne, M, Smith, DL (2012) ``One-size-fits-all''? Optimizing treatment duration for bacterial infections. PloS ONE 7:e29838. https://doi.org/10.1371/journal.pone.0029838
	
	\bibitem[Goranova et al. (2022)]{goranova}
	Goranova, M, Ochoa, G, Maier, P, Hoyle, A (2022) Evolutionary optimisation of antimicrobial dosing regimens for bacteria with different levels of resistance. Artificial Artif. Intell. Med. 133:102405. https://doi.org/10.1016/j.artmed.2022.102405
	
	\bibitem[Hoyle et al. (2020)]{hoyle}
	Hoyle, A, Cairns, D, Paterson, I, McMillan, S, Ochoa, G, Desbois, AP (2020) Optimising efficacy of antimicrobials against systemic infection by varying dosage quantities and times. PLoS Comput. Biol. 16:e1008037. https://doi.org/10.1371/journal.pcbi.1008037
	
	
	\bibitem[Kesisoglou et al (2022)]{kesisoglou}
	Kesisoglou, I, Tam, VH, Tomaras, AP, Nikolaou, M (2022) Discerning in vitro pharmacodynamics from OD measurements: A model-based approach. Comput. Chem. Eng. 158:107617. https://doi.org/10.1016/j.compchemeng.2021.107617
	
	\bibitem[Khan and Imran (2018)]{khan}
	Khan, A, Imran, M (2018) Optimal dosing strategies against susceptible and resistant bacteria. J. Biol. Syst. 26:41-58. https://doi.org/10.1142/S0218339018500031
	
	\bibitem[Krzyzanski and Jusko (1998)]{krzyzanski}
	Krzyzanski, W, Jusko, WJ (1998). Integrated functions for four basic models of indirect pharmacodynamic response. J. Pharm. Sci. 87:67-72.
	https://doi.org/10.1021/js970168r

\bibitem[Ledzewicz and Sch\"attler (2021)]{Ledzewicz}
Ledzewicz, U, Sch\"attler, H (2021) On the role of pharmacometrics in mathematical models for cancer treatments. Discrete Continuous Dyn Syst Ser B 26:483-499.
https://doi.org/10.3934/dcdsb.2020213
	
\bibitem[Leszczy\'nski et al. (2020)]{Leszczynski}
Leszczy\'nski, M, Ledzewicz, U, Sch\"attler, H (2020) Optimal control for a mathematical model for chemotherapy with pharmacometrics. Math. Modell. Nat. Phenom. 15:69. 	https://doi.org/10.1051/mmnp/2020008
	
	\bibitem[Levin and Udekwu (2010)]{levin}
	Levin, BR, Udekwu, KI (2010) Population dynamics of antibiotic treatment: a mathematical model and hypotheses for time-kill and continuous-culture experiments. Antimicrob. Agents Chemother. 54:3414-3426.  https://doi.org/10.1128/AAC.00381-10
	
	\bibitem[Lipsitch and Levin (1997)]{lipsitch}
	Lipsitch, M, Levin, BR (1997) The population dynamics of antimicrobial chemotherapy. Antimicrob. Agents Chemother. 41: 363-373.  https://doi.org/10.1128/AAC.41.2.363
	
	\bibitem[Luterbach and Rao (2022)]{luterbach}
	Luterbach, CL,  Rao, GG (2022) Use of pharmacokinetic/pharmacodynamic approaches for dose optimization: a case study of plazomicin. Curr. Opin. Microbiol. 70:102204.
	https://doi.org/10.1016/j.mib.2022.102204
	
	\bibitem[Machera and Iliadis (2016)]{macheras}
	Macheras, P, Iliadis, A (2016) Modeling in biopharmaceutics, pharmacokinetics and pharmacodynamics: homogeneous and heterogeneous approaches, Springer, Heidelberg.
	
	\bibitem[Marrec and Bitbol (2020)]{marrec}
	Marrec, L, Bitbol, AF (2020) Resist or perish: fate of a microbial population subjected to a periodic presence of antimicrobial. PLoS Comput. Biol. 16:e1007798.
	https://doi.org/10.1371/journal.pcbi.1007798
	
	\bibitem[Meibohm and Derendorf (1997)]{meibohm}
	Meibohm, B, Derendorf, H (1997) Basic concepts of pharmacokinetic/pharmacodynamic (PK/PD) modelling. Int. J. Clin. Pharmacol. Ther. 35:401-413.
	
	\bibitem[Mi et al. (2022)]{mi}
	Mi, K, Zhou, K, Sun, L,
	Hou, Y, Ma, W, Xu, X, Huo, M, Liu,
	Z, Huang, L. (2022) Application of
	Semi-Mechanistic Pharmacokinetic
	and Pharmacodynamic Model in
	Antimicrobial Resistance.
	Pharmaceutics 14,246.
	https://doi.org/10.3390/pharmaceutics14020246
	
	\bibitem[Morsky and Vural (2022)]{morsky}
	Morsky, B, Vural, DC (2022) Suppressing evolution of antibiotic resistance through environmental switching. Theor. Ecol. 15:115-127. https://doi.org/10.1007/s12080-022-00530-4
	
	\bibitem[Mouton and Vinks (2005)]{mouton}
	Mouton, JW, Vinks, AS (2005) Pharmacokinetic/Pharmacodynamic modelling of antibiotics in vitro and in vivo using bacterial growth and kill kinetics: the zMIC vs stationary concentrations. Clin. Pharmacokinet. 44: 201-10. https://doi.org/10.2165/00003088-200544020-00005
	
	\bibitem[Mudassar and Smith (2005)]{mudassar}
	Mudassar, I, Smith, H (2005) The pharmacodynamics of antibiotic treatment. Comput. Math. Methods. Med., 7:229-263.
	https://doi.org/10.1080/10273660601122773
	
	\bibitem[Mueller et al. (2004)]{mueller}
	Mueller, M, de la Pena, A, Derendorf, H (2004) Issues in pharmacokinetics and pharmacodynamics of anti-infective agents: kill curves versus MIC. Antimicrob. Agents Chemother. 48:369-377.
	 https://doi.org/10.1128/AAC.48.2.369-377.2004
	
	
	\bibitem[Murray et al. (2022)]{murray}
	Murray, CJ, Ikuta, KS, Sharara, F, et al (2022) Global burden of bacterial antimicrobial resistance in 2019: a systematic analysis. The Lancet 399:629-655.
	https://doi.org/10.1016/S0140-6736(21)02724-0
	
	\bibitem[Nguyen and Peletier (2009)]{nguyen}
	Nguyen, HM, Peletier, LA (2009), Monotonicity of time to peak response with respect to drug dose for turnover models, Diff. Int. Eq. 22:1-26.
	
	\bibitem[Nielsen and Friberg (2013)]{nielsen}
	Nielsen, EI, Friberg, LE (2013) Pharmacokinetic-pharmacodynamic modeling of antibacterial drugs. Pharmacol. Rev. 65:1053-1090. https://doi.org/10.1124/pr.111.005769
	
	\bibitem[Nikolaou and Tam (2006)]{nikolaou1}
	Nikolaou, M, Tam, VH (2006) A new modeling approach to the effect of antimicrobial agents on heterogeneous microbial populations. J. Math. Biol. 52:154-182. https://doi.org/10.1007/s00285-005-0350-6
	
	\bibitem[Nikolaou et al. (2007)]{nikolaou2}
	Nikolaou, M, Schilling, AN, Vo, G, Chang, KT, Tam, VH (2007) Modeling of microbial population responses to time-periodic concentrations of antimicrobial agents. Ann. Biomed. Eng. 35:1458-1470. https://doi.org/10.1007/s10439-007-9306-x
	
	\bibitem[Onufrak et al. (2016)]{onufrak}
	Onufrak, NJ, Forrest, A, Gonzalez, D (2016) Pharmacokinetic and pharmacodynamic principles of anti-infective dosing. Clin. Ther. 38:1930-1947. https://doi.org/10.1016/j.clinthera.2016.06.015
	
	\bibitem[Owens et al. (2004)]{owens}
	Owens, RC, Nightingale, CH, Ambrose, PG (Eds.) (2004) Antibiotic optimization: concepts and strategies in clinical practice. Marcel Dekker, New York.
	
	\bibitem[Paterson et al. (2016)]{paterson}
	Paterson, IK, Hoyle, A, Ochoa, G, Baker-Austin, C, Taylor, NG (2016) Optimising antimicrobial usage to treat bacterial infections. Sci. Rep. 6:1-10. https://doi.org/10.1038/srep37853
	
	\bibitem[Peletier et al. (2005)]{peletier}
	Peletier, LA, Gabrielsson, J, Haag, JD (2005). A dynamical systems analysis of the indirect response model with special emphasis on time to peak response. J. Pharmacokinet. Pharmacodyn. 32:607-654.
	https://doi.org/10.1007/s10928-005-0047-x
	
	\bibitem[Pe\~na-Miller et al. (2012)]{pena}
	Pe\~na-Miller, R, L\"ahnemann, D, Schulenburg, H, Ackermann, M, \& Beardmore, R (2012) Selecting against antibiotic-resistant pathogens: optimal treatments in the presence of commensal bacteria. Bull. Math. Biol. 74:908-934. https://doi.org/10.1007/s11538-011-9698-5
	
	\bibitem[Rayner et al. (2021)]{rayner}
	Rayner, CR, Smith, PF et. al. (2021) Model informed drug development for anti?infectives: state of the art and future. Clin. Pharmacol. Ther. 109:867-891.
	https://doi.org/10.1002/cpt.2198
	
	\bibitem[Rao and Landersdorfer (2021)]{rao}
	Rao, GG, Landersdorfer, CB (2021) Antibiotic pharmacokinetic/pharmacodynamic modelling: zMIC, pharmacodynamic indices and beyond. Int. J. Antimicrob. Agents 58:106368.
	https://doi.org/10.1016/j.ijantimicag.2021.106368
	
	\bibitem[Regoes et al. (2004)]{regoes}
	Regoes, RR, Wiuff, C, Zappala, RM, Garner, KN, Baquero, F., \& Levin, B.R. (2004) Pharmacodynamic functions: a multiparameter approach to the design of antimicrobial treatment regimens. Antimicrob. Agents Chemother. 48:3670-3676.
	https://doi.org/10.1128/AAC.48.10.3670-3676.2004
	
	\bibitem[Rescigno (2003)]{rescigno}
	Rescigno, A. (2003) Foundations of Pharmacokinetics. Springer, New York.
	
	\bibitem[Rotschafer et al. (2016)]{rotschafer}
	Rotschafer, JC, Andes, DR, Rodvold, KA (Eds.) (2016) Antibiotic Pharmacodynamics. Springer, New York.
	
	\bibitem[Shi et al. (2014)]{shi}
Shi, J, Alagoz, O, Erenay, FS, Su, Q (2014) A survey of optimization models on cancer chemotherapy treatment planning. Ann. Oper. Res. 221:331-356.
https://doi.org/10.1007/s10479-011-0869-4
	
	\bibitem[Singh et al. (2023)]{singh}
	Singh, G, Orman, MA, Conrad, JC, Nikolaou, M. (2023) Systematic design of pulse dosing to eradicate persister bacteria. PLoS Comput. Biol. 19:e1010243. https://doi.org/10.1371/journal.pcbi.1010243
	
	\bibitem[Smith et al. (2020)]{smith}
	Smith, NM, Lenhard, JR et. al. (2020) Using machine learning to optimize antimicrobial combinations: dosing strategies for meropenem and polymyxin B against carbapenem-resistant Acinetobacter baumannii. Clin. Microbiol Infecti. 26:1207-1213.
	https://doi.org/10.1016/j.cmi.2020.02.004
	
	\bibitem[Tindall et al. (2022)]{tindall}
	Tindall, M, Chappell, MJ, Yates, JW (2022) The ingredients for an antimicrobial mathematical modelling broth. Int. J. Antimicrob. Agents 60:106641.
	https://doi.org/10.1016/j.ijantimicag.2022.106641
	
	\bibitem[Ventola (2015)]{ventola}
	Ventola, CL (2015) The antibiotic resistance crisis: part 1: causes and threats. Pharm. Ther. 40:277.
	
	\bibitem[Vinks et al. (2014)]{vinks}
	Vinks, AA, Derendorf, H, Mouton, JW (Eds.) (2014) Fundamentals of antimicrobial pharmacokinetics and pharmacodynamics. Springer, New York.
	
	\bibitem[Wen et al. (2016)]{wen}
	Wen, X, Gehring, R, Stallbaumer, A, Riviere, JE, Volkova, V. V. (2016) Limitations of zMIC as sole metric of pharmacodynamic response across the range of antimicrobial susceptibilities within a single bacterial species. Sci. Rep. 6:1-8. https://doi.org/10.1038/srep37907
	
	\bibitem[Wu et al. (2022)]{wu}
	Wu, X, Zhang, H, Li, J (2022). An Analytical Approach of One-Compartmental Pharmacokinetic Models with Sigmoidal Hill Elimination. Bull. Math. Biol., 84:117.
	https://doi.org/10.1007/s11538-022-01078-4
	
	\bibitem[Zhi et al. (1988)]{zhi}
	Zhi, J, Nightingale, CH, Quintiliani, R (1988) Microbial pharmacodynamics of piperacillin in neutropenic mice of systematic infection due to Pseudomonas aeruginosa. J. pharmacokinetic. biopharm. 16:355-375. https://doi.org/10.1007/BF01062551
	
	\bibitem[Zilonova et al. (2016)]{zilonova}
	Zilonova, EM, Bratus, AS (2016) Optimal strategies in antibiotic treatment of microbial populations. Appl. Anal., 95:1534-1547. https://doi.org/10.1080/00036811.2016.1143552
	
	
\end{thebibliography}
\end{document}